\documentclass[12pt,a4paper]{article}

\usepackage[T1]{fontenc}

\usepackage{mathrsfs,dsfont}
\usepackage{bbm}
\usepackage{amssymb} 
\usepackage{mathtools}

\usepackage{geometry}
\usepackage{graphicx}   
\usepackage{enumerate}  
\usepackage{tabularx}   
\usepackage{booktabs}   
\usepackage[labelfont=bf,font=small,width=.95\textwidth]{caption} 

\usepackage[dvipsnames]{xcolor}
\definecolor{linkred}{HTML}{C70039}
\usepackage[colorlinks]{hyperref}
\hypersetup{
    linkcolor=linkred, 
    citecolor=RoyalBlue,
}

\usepackage{amsthm}
\newtheorem{thm}{Theorem}
\newtheorem{cor}[thm]{Corollary}
\newtheorem{lem}[thm]{Lemma}
\newtheorem{prop}[thm]{Proposition}
\newtheorem{defn}[thm]{Definition}
\newtheorem{rem}{Remark}

\newcommand{\diff}{\mathrm{d}}

\DeclareMathOperator{\rank}{rank}
\DeclareMathOperator{\Id}{Id}

\renewcommand{\P}{\mathbb{P}}
\newcommand{\E}{\mathbb{E}}
\newcommand{\R}{\mathbb{R}}

\newcommand{\icu}{\mathrm{ICU}}
\newcommand{\hosp}{\mathrm{hosp}}

\newcommand{\norm}[1]{\lVert #1 \rVert}
\newcommand{\abs}[1]{\lvert #1 \rvert}
\DeclarePairedDelimiter{\floor}{\lfloor}{\rfloor}
\DeclarePairedDelimiter{\Angle}{\langle}{\rangle}

\title{
From individual-based epidemic models to McKendrick-von Foerster PDEs:
A guide to modeling and inferring COVID-19 dynamics  
}

\usepackage{authblk}

\newlength{\affilskip}
\author[1]{Félix Foutel-Rodier\footnote{First author.}}
\author[2,3]{François Blanquart}
\author[2]{Philibert Courau}
\author[2,4]{Peter~Czuppon}
\author[5]{Jean-Jil Duchamps}
\author[2]{Jasmine Gamblin}
\author[2,6]{\'Elise~Kerdoncuff}
\author[2,7]{Rob Kulathinal}
\author[2]{L\'eo R\'egnier}
\author[2]{Laura Vuduc}
\author[2,8]{Amaury Lambert\footnote{Co-last authors.}} 
\author[9]{Emmanuel Schertzer\textsuperscript{\textdagger}}

\affil[1]{
    Département de Mathématiques, Université du Québec à Montréal,\authorcr
    Montréal, QC, Canada\vspace{\affilskip}
}
\affil[2]{
    SMILE Group, Center for Interdisciplinary Research in Biology UMR 7241,\authorcr
    Coll\`ege de France, CNRS, INSERM U 1050, \authorcr
    PSL Research University, Paris, France \vspace{\affilskip}
}

\affil[3]{
    Infection, Antimicrobials, Modeling, Evolution UMR 1137, \authorcr
    Université de Paris, INSERM, Paris, France
    \vspace{\affilskip}
}
\affil[4]{
    Institute for Evolution and Biodiversity, University of Münster, 48149
    Münster, Germany
    \vspace{\affilskip}
}
\affil[5]{
    Laboratoire de mathématiques de Besan{\c c}on UMR 6623, \authorcr
    Université Bourgogne Franche-Comté, CNRS, F-25000 Besançon, France
    \vspace{\affilskip}
}
\affil[6]{
    Institut de Systématique, Biodiversité, \'Evolution UMR 7205, \authorcr 
    Muséum National d'Histoire Naturelle, CNRS, Paris, France
    \vspace{\affilskip}
}
\affil[7]{
    Department of Biology, Temple University, Philadelphia, PA, USA
    \vspace{\affilskip}
}
\affil[8]{
    Institut de Biologie de l'ENS, \'Ecole Normale Supérieure, \authorcr
    CNRS UMR 8197 INSERM U 1024, Paris, Franc
    \vspace{\affilskip}
}
\affil[9]{
    Faculty of Mathematics, University of Vienna, \authorcr
    Oskar-Morgenstern-Platz 1, 1090 Wien, Austria 
}
\date{}

\begin{document}

\maketitle

\vspace*{-1cm}

\begin{abstract}
We present a unifying, tractable approach for studying the spread
of viruses causing complex diseases requiring to be modeled using a
large number of types (e.g., infective stage, clinical state, risk
factor class). We show that recording each infected individual's
infection age, i.e., the time elapsed since infection, has three
benefits.

First, regardless of the number of types, the age distribution of the
population can be described by means of a first-order, one-dimensional
partial differential equation (PDE) known as the McKendrick-von Foerster
equation. The frequency of type $i$ is simply obtained by integrating the
probability of being in state $i$ at a given age against the age
distribution. This representation  induces a simple methodology based on
the additional assumption of Poisson sampling to infer and forecast the
epidemic. We illustrate this technique using French data from the
COVID-19 epidemic.

Second, our approach generalizes and simplifies standard compartmental
models using high-dimensional systems of ordinary differential equations
(ODEs) to account for disease complexity. We show that such models can
always be rewritten in our framework, thus, providing a low-dimensional
yet equivalent representation of these complex models. 

Third, beyond the simplicity of the approach, we show that our population
model naturally appears as a universal scaling limit of a large class of
fully stochastic individual-based epidemic models,  where the initial
condition of the PDE emerges as the limiting age structure of an
exponentially growing population starting from a single individual.
\end{abstract}

\section{Introduction}

\subsection{Challenges posed by complex diseases such as COVID-19}

The transmission of pathogens between species is a global concern
\cite{Krauss2003,Crawford}. As such zoonotic episodes are expected to
become increasingly common in humans, it is critical to develop analytic
tools that can quickly transform epidemiological observations into
informed public policy in order to mitigate and control epidemics. 

A novel coronavirus, SARS-CoV-2, has recently crossed the species barrier
into humans and, within months, has rapidly spread to all corners of our
planet \cite{Wu2020}. The sheer scale of this pandemic has overburdened
our medical infrastructure, caused fatalities estimated well into
millions, and shut down entire economies. Remarkably, the
rapid spread of COVID-19 and its consequences can be attributed to the
unique life cycle of a 30,000 base pair single-stranded virus. SARS-CoV-2
is an airborne pathogen transmitted by both symptomatic and asymptomatic
carriers in close proximity to non-infected individuals. Milder COVID-19
symptoms include a dry cough, fever, and/or shortness of breath while
more serious cases include respiratory failure and possible death. With
millions of infections and hundreds of thousands of documented deaths and
recoveries, the COVID-19 pandemic is providing a wealth of independent
estimates of important clinical characteristics that can help predict
health outcomes specific for a country or region.

It quickly became understood that accurate descriptions of the life cycle
of this disease needed to distinguish between several stages of the
disease, referred to as compartments, depending on whether an infected
individual is infectious or not, symptomatic or not, hospitalized,
\textit{etc}. However it remains unclear to what extent making precise
predictions of the dynamics of such a complex disease requires to have a
precise knowledge of clinical features such as incubation period,
generation time, and duration times between infection, symptom
establishment, hospitalization, recovery and death, to know how these
durations correlate and what are the exact probabilities of transition
between stages.

In this work, we consider a fully stochastic, generic epidemiological
model with an arbitrary number of compartments, that encompasses life
cycles of most complex diseases and that of COVID-19 in particular. We
show how structuring the infected population by its infection age, i.e.,
time elapsed since infection, allows us to decouple dependencies between
stages and to time. More specifically, when the population size is large
enough, the joint evolution of all compartment sizes can be described by
means of a linear, first-order partial differential equation (PDE) known
as the McKendrick-von Foerster equation describing the number $n(t,a)$ of
infecteds of (infection) age $a$ at time $t$. The boundary condition at
age 0 is driven by the infection rate from infecteds of age $a$,
\emph{averaged} over all possible courses of infection, and the number of
individuals of age $a$ in compartment $i$ at time $t$ is obtained by
thinning $n(t,a)$ by a factor $p(a,i)$ which is the probability of being
in compartment $i$ conditional on having age $a$, \emph{averaged} over
all possible courses of infection.

In the case of COVID-19, we display a simple procedure to infer these
parameters, some from the biological literature and most from time series
of numbers of severe cases, hospitalized cases, discharged patients and
deaths that can be applied easily to any regional or national dataset. We
also allow for time inhomogeneity in the infection rate to account for
temporary mitigation measures such as lockdowns or social distancing. We
apply this procedure to French COVID-19 data from March to May 2020 and
estimate various parameters of interest including the reproduction number
in different phases of the epidemic (before, during, and after lockdown)
and biological parameter values that we compare to empirical estimates.

\subsection{Decorated age-structured epidemic models} \label{SS:SIR}

The large population size limit of our stochastic model is a PDE
``decorated'' with compartments. This point of view extends the usual
sets of ODEs used in epidemiology, and allows us to represent in the same
framework a large class of deterministic epidemic models. Before
describing the stochastic model underlying such decorated PDEs, let us
illustrate this notion by recalling the well-known derivation of the
classic SIR set of ODEs from an age-structured model. 

Consider the solution $(n(t,a);\, t,a \ge 0)$ to the following partial
differential equation:
\begin{equation} \label{eq:Kermack_McK}
    \begin{aligned}
    \partial_t n + \partial_a n &= 0 \\
    \forall t \ge 0,\; n(t,0)   &= S(t) \int_0^\infty n(t,a) \tau(a) \diff a \\
    \forall a \ge 0,\; n(0,a)   &= x_0 g(a) \\
    \forall t \ge 0,\; S(t) &= 1 - \int_0^\infty n(t,a) \diff a.
    \end{aligned}
\end{equation}
where $0 \le x_0 \le 1$, and $g, \tau \ge 0$ fulfill
\[
    \int_0^\infty g(a) \diff a = 1,\quad \int_0^\infty \tau(a) \diff a <
    \infty.
\]
Equation~\eqref{eq:Kermack_McK} was first proposed to describe the dynamics
of an epidemic where infected individuals are structured by their
\emph{age of infection}, and is known as the Kermack-McKendrick model \cite{kermack1927}. 
Note that in the original work of \cite{kermack1927} the model is
formulated as the solution to a convolution equation rather than as a PDE, but
that the two formulations are equivalent, see Section~\ref{SS:weakSolution}.
In this context, the age of an individual refers to the time elapsed
since its infection, and not to its actual age. Then, $n(t,a)$ is the
density at time $t$ of all individuals with age (of infection) $a$, and
$S(t)$ the density of individuals still susceptible to the disease. The 
differential term describes the aging process: the age of an individual
increases linearly with time at rate one. The interpretation of the age
boundary condition of \eqref{eq:Kermack_McK} is that individuals with age
$a$ infect susceptible individuals at a rate $\tau(a)$ that only depends
on their age.

\begin{rem}
    It is important to note that $n(t,a)$ counts all individuals that
    have been infected at time $t-a$, and not only \emph{infective}
    individuals with age $a$. Thus, Eq.~\eqref{eq:Kermack_McK} lacks the
    usual recovery term. Moreover, $\tau(a)$ is not the average rate at
    which \emph{infective} individuals with age $a$ yield infections, but
    the average infection rate of any individual with age $a$. (The
    former is obtained from the latter by discounting all individuals
    with age $a$ that are not infectious anymore.)
\end{rem}

In order to recover the SIR model, suppose that $\tau$ is given by
\[
    \forall a \ge 0,\quad \tau(a) = \beta e^{-\gamma a},
\]
for some $\gamma, \beta > 0$. Further define
\[
    \forall t \ge 0,\quad I(t) = \int_0^\infty n(t,a) e^{-\gamma a} \diff a,
    \quad
    R(t) = \int_0^\infty n(t,a) (1-e^{-\gamma a}) \diff a,
\]
Then a simple calculation shows that $(S,I,R)$ solves the following
well-known system of ODEs:
\begin{equation} \label{eq:SIR}
    \begin{aligned}
    \dot{S} &= -\beta I S\\
    \dot{I} &= \beta I S - \gamma I\\
    \dot{R} &= -\gamma I.
    \end{aligned}
\end{equation}
The previous expressions have an interesting probabilistic
interpretation. Consider a Markov process $(X(a);\, a \ge 0)$ with two
states $I$ and $R$. Suppose that it starts from $I$ and jumps to $R$ at
rate $\gamma$. The process $(X(a);\, a \ge 0)$ can be interpreted as
describing the sequence of states (infective then recovered) visited by a
typical individual in the microscopic model underlying \eqref{eq:SIR}.
Then, clearly 
\[
    p(a,I) \coloneqq \P(X(a) = I) = e^{-\gamma a},\quad 
    p(a,R) \coloneqq \P(X(a) = R) = 1-e^{-\gamma a},
\]
so that
\begin{equation} \label{eq:SIRcomp}
    I(t) = \int_0^\infty n(t,a) p(a, I) \diff a, \quad
    R(t) = \int_0^\infty n(t,a) p(a, R) \diff a.
\end{equation}
Furthermore, suppose that a typical infected individual yields new
infections at constant rate $\beta$ while it is in state $I$. Then, the
mean number of new infections occurring in the time interval $[a, a+\diff
a]$ is
\[
    \beta e^{-\gamma a} \diff a = \tau(a) \diff a.
\]

The picture that emerges from this simple calculation is that, instead of
keeping track of the number of individuals in each compartment, one can
consider the age structure of the population, given by 
Eq.~\eqref{eq:Kermack_McK}. The dynamics of the age structure is uniquely
prescribed by the average number $\tau(a)$ of infections that an
individual yields at age $a$. The individual counts in each compartment
can then be recovered by integrating against the age structure the
one-dimensional marginals of a process $(X(a);\, a \ge 0)$ that describes
the sequence of compartments visited by a typical individual in the
population. We say that the PDE is ``decorated'' with compartments, as
the process $(X(a);\, a \ge 0)$ is used to recover the counts in the
compartments, and only influences the dynamics of the infection
which is described by the sole Eq.~\eqref{eq:Kermack_McK} through the 
average infection rate $\tau(a)$.

This viewpoint has several advantages compared to the usual ODE setting
of \eqref{eq:SIR}. First, we can make sense of \eqref{eq:SIRcomp} for any
process $(X(a);\, a \ge 0)$. If this process is not Markovian, the number
of individuals in each compartment no longer solves a system of ODEs
similar to \eqref{eq:SIR}. Hence, this approach allows to go beyond the
usual ODE framework. This generalization is of great modeling interest
since a hypothesis underlying sets of ODEs is that the sojourn times in
each compartment and the time between successive infections are
exponentially distributed. In particular they cannot account for sojourn
time distributions that are peaked around a value, which have been
reported for instance for COVID-19 \cite{Linton2020,
verity_estimates_2020}.

Second, regardless of the number of compartments, the age structure of
the population is described by the same one-dimensional PDE. This is
particularly valuable when models have a large number of compartments, as
in the context of COVID-19 \cite{salje_estimating_2020,
evgeniou_epidemic_2020, DiDomenico2020, djidjou_optimal_2020}, as it
avoids the use of high-dimensional systems of ODEs that are cumbersome to study
mathematically. However, this requires to work with the PDE
\eqref{eq:Kermack_McK} rather than with ODEs, resulting in a mild
computational cost.

Third, Eq.~\eqref{eq:Kermack_McK} only involves the mean number of
infections $\tau(a)$ induced by individuals at age $a$ \emph{averaged
over all compartments}. In particular, it is unnecessary to assess to
which compartments individuals belong when they yield new infections.
This is in contrast with the usual ODE framework, where for each
compartment an infection rate needs to be prescribed. As we will see,
$\tau(a)$ relates to well-known epidemiological quantities that can be
assessed directly from the data.

The main contribution of our work is to show that such decorated
age-structured models arise naturally as the law of large numbers limit
of a wide class of general stochastic epidemic models that we now
introduce.

\subsection{Generic stochastic model assumptions}
\label{SS:introModel}

We consider a population model in which individuals are either
susceptible, if they have not yet met the disease, or infected. Our
definition of infected is broader than usual: an individual is infected
if it has been infected in the past. In particular, individuals that
have recovered or died from the disease are still infected, even if they
are not infective anymore. At any point in time, an infected individual
is in one of several states, that will also be referred to as
compartments, types, classes, or stages. The set of all such states is
denoted by $\mathcal{S}$ and is assumed to be finite. Depending on the
disease complexity, the number of stages can vary. In the SARS-CoV-2
example, typical stages are asymptomatic, mild case, severe case,
hospitalized, intensive care unit, recovered, and dead (see
Figure~\ref{F:admissionModel}).

We assume that upon infection a susceptible individual immediately
changes state and never becomes susceptible anymore (ruling out multiple
infections, in particular), and that it will eventually end in one of two
states: recovered or dead. The sequence of states visited by an
individual $x$ is then encoded by a stochastic process $X_x \coloneqq
(X_x(a);\, a \ge 0)$ valued in $\mathcal{S}$, where the random variable
$X_x(a)$ is the state of $x$ at age of infection $a$. We call $(X_x(a);\,
a \ge 0)$ the \emph{life-cycle process}.

Each individual is further endowed with a random point process
$\mathcal{P}_x$ on $[0, \infty)$, called the \emph{infection point
    process}. Each atom of $\mathcal{P}_x$ gives the age at which
$x$ makes an infectious contact with another individual in the
population, and we assume that all atoms are distinct so that
secondary infections cannot occur simultaneously. (By an infectious
contact, we mean a contact that would lead to a new infection if the
target individual is susceptible to the disease. The pair
$(\mathcal{P}_x, X_x)$ characterizes the course of the infection of
individual $x$. We assume that these pairs, for all individuals $x$, are
i.i.d.\ copies of the same pair $(\mathcal{P}, X)$ that describes the
infection of a typical individual in the population.

In order to make the mathematical treatment of the model easier, we make
the simplifying assumption that the number of susceptible individuals is
large compared to the number of infected individuals, so that each
infectious contact leads to a new infection. (We neglect the saturation
in the number of infections due to the finiteness of the population.)
The epidemic is then described by a branching process known as a
\emph{Crump-Mode-Jagers process} \cite{jagers1975branching, Taib1992}: an
individual $x$ infected at time $\sigma_x$ will produce $N_x$ new
secondary infections at times $\sigma_x + A_1, \dots, \sigma_x +
A_{N_x}$, where $(A_1, \dots, A_{N_x})$ are the atoms of $\mathcal{P}_x$,
that is ,
\[
    \mathcal{P}_x = \sum_{i = 1}^{N_x} \delta_{A_i}.
\]
(This branching hypothesis is relaxed in a recent work by some of
the authors \cite{duchamps2021general}.)

Lastly, we superimpose time heterogeneity to this process by means of a
\emph{contact rate} $(c(t);\, t \ge 0)$ valued in $[0, 1]$
thinning the infection process. More precisely, if $t$ is a potential
time of infection for individual $x$, we ignore the infection with
probability $1-c(t)$. This contact rate can model the effect of
vaccination, or density-dependence (i.e., relaxing the branching
assumption due to an excess of removed or of deceased individuals), or of
governmental mitigation measures (i.e., social distancing, lockdown).

The infection process is more formally constructed in
Section~\ref{SS:assumption}.

\begin{rem}
    As already discussed in the previous section, in the SIR example
    $\mathcal{S} = \{I, R\}$, and $(X(a);\, a \ge 0)$ is a Markov process
    started from $I$ that jumps to $R$ at rate $\gamma$. Moreover, the
    infection point process is given by
    \[
        \mathcal{P} = \sum_{A_i \in P} \delta_{A_i} \mathbbm{1}_{\{X(A_i) = I\}}
    \]
    where $P$ is a homogeneous Poisson point process on $[0, \infty)$
    with intensity $\beta$.
\end{rem}

\begin{rem}
    We emphasize that the \emph{pairs} $(\mathcal{P}_x, X_x)$ are assumed
    to be independent, but \emph{not} the variables $\mathcal{P}_x$ and $X_x$.
    In the simple SIR example they are not independent since there can be
    no atoms of $\mathcal{P}_x$ after the recovery time. In the same
    spirit, one could assume that the infection potential of a given
    individual is reduced once in the hospital and that individuals with
    many atoms in their infection process $\mathcal{P}_x$ (high
    infectiosity) are identified and isolated.
\end{rem}

\subsection{Statement of the main results} \label{SS:mainResults}

The stochastic epidemic models we consider here are fairly general and
can exhibit quite complex dependencies (i) between states and time, due
to the lack of any Markov-type assumption, (ii) between states, due to
possibly hidden structuring variables impacting the life cycle, (iii)
between state and infection rate, and (iv) between past and future
infection events. The main result of this work is that despite this
apparent complexity, most of this complexity vanishes when the size of
the population is large. More specifically, we show that in the limit of
large populations (obtained by starting from a large initial population
or as a consequence of natural exponential growth), the population of
infected individuals structured by age (of the infection) can be described
by means of a one-dimensional PDE, and that the counts in each compartment  
are recovered by decorating this PDE with the life-cycle process. The
limiting expression only depends on:
\begin{enumerate}
    \item The average infection rate 
    \[
        \tau(\diff a) \coloneqq \mathbb{E}(\mathcal{P}(\diff a)),
    \]
    formally defined as the intensity measure of $\mathcal{P}$.
    We make the simplifying assumption that $\tau$ has a density 
    w.r.t.\ the Lebesgue measure and, with a slight abuse of notation,  
    we still denote it by $\tau(a)$.
    \item The one-dimensional marginals of the life-cycle process 
    \[
        p(a,i) \coloneqq \mathbb{P}(X(a)=i). 
    \]
\end{enumerate}
We prove two main theorems that are two laws of large numbers for the
age and compartment structure in the population: one started from a large
number of individuals, and the other from a single individual.

\paragraph{Large initial population.} Let us start the population with a
large number $N$ of infected individuals at time $t=0$, with i.i.d.\
initial infection ages with law $g$. (See Section~\ref{SS:assumption}
for a formal definition of the initial age.) Define the empirical measure
of ages and compartments at time $t$ as 
\begin{align}\label{def:empirical}
    \mu^N_t(\diff a \times \{i\})  \coloneqq   
    \sum_{\sigma_x < t} \delta_{(t-\sigma_x, X_x(t-\sigma_x))}(\diff a\times \{i\}),
\end{align}
where $\sigma_x$ denotes the infection time of $x$, and the sum is taken
over all individuals $x$ infected before time $t$. (This includes
the initial individuals infected before time $0$.) The measure $\mu_t^N$
is a random point measure that encodes the ages and compartments of all
individuals that have been infected before time $t$.

Let us also introduce $n^N_i(t)$, the number of individuals in
compartment $i$ at time $t$, defined as 
\[
    n^N_i(t) \coloneqq \sum_{\sigma_x < t} \mathbbm{1}_{\{ X_x(t-\sigma_x)=i\} }
    = \mu_t^N\big([0, \infty) \times \{i\}\big).
\]

\begin{thm}[$N$ individuals] \label{thm:largePop}
    Start the population with $N$ individuals with i.i.d.\ initial ages
    distributed according to $g$. Then, for any $t > 0$, the following
    convergence holds in the weak topology
    \[
        \frac{1}{N} \mu^N_t(\diff a \times \{i\}) \underset{N \to \infty}{\longrightarrow} 
        n(t,a) p(a,i) \diff a \quad \text{a.s.}
    \]
    where $(n(t,a);\, t,a \ge 0)$ is the solution to
    \begin{equation} \label{eq:McKVF}
        \begin{aligned}
        \partial_t n + \partial_a n &= 0 \\
        \forall t \ge 0,\; n(t,0) &= c(t) \int_0^\infty n(t,a) \tau(a) \diff a \\
        \forall a \ge 0,\; n(0,a) &= g(a).
        \end{aligned}
    \end{equation}
    As a consequence, for any $t > 0$,
    \begin{equation} \label{eq:compartmentConv}
        \frac{1}{N} n^N_i(t) \underset{N \to \infty}{\longrightarrow}
        \int_0^\infty n(t,a) p(a,i) \diff a \quad\text{a.s.}
    \end{equation}
\end{thm}

The limiting age structure of the population is thus described by Eq.\
\eqref{eq:McKVF}, which is a linear version of Eq.\
\eqref{eq:Kermack_McK}, known as the McKendrick-von~Foerster equation.
Note that it also has an additional $c(t)$ term accounting for the
reduced contact rate, resulting in a time heterogeneity. As in
Section~\ref{SS:SIR}, the number of individuals in each compartment is
recovered by decorating the PDE with the one-dimensional marginals of the
life-cycle process.

\paragraph{After lockdown onset.}
Our second result displays a similar, but more subtle, convergence in the
case when the process is supercritical, where natural growth leads by
itself to large population sizes. We say that the process is
supercritical if 
\[
    \int_0^\infty \tau(a) \diff a > 1,
\]
in which case there exists $\alpha > 0$ such that 
\[
    \int_0^\infty e^{-\alpha a} \tau(a) \diff a = 1.
\]
(The parameter $\alpha$ is the \emph{Malthusian parameter} of the CMJ
process when $c \equiv 1$.) Let $Z(t)$ denote the total population size
at time $t$ and assume that $Z(0)=1$, i.e., we start from \emph{a single
individual}. Suppose that $(t_K;\, K \ge 0)$ is a sequence of stopping
times (with respect to the natural filtration of the process, see
Section~\ref{S:proofs}) such that $t_K \to \infty$ on the non-extinction
event. By a slight abuse of notation, denote by $\mu^K_t$ the empirical
measure of ages and types as in \eqref{def:empirical}, but under the
assumption that the contact rate at time $t$ is equal to
$c(t-t_K)$ where $c$ is equal to 1 for negative arguments. We are
motivated by modeling a situation where the infection is separated into
two distinct phases: 
\begin{enumerate}
   \item The epidemic develops until a certain random time $t_K$. For
       instance, $t_K$ could be the time at which the number of recorded
       deaths exceeds a large threshold $K$. We assume no suppression
       before $t_K$.
   \item We let the contact rate vary after time $t_K$ according to the
       function $(c(t-t_K);\, t \ge 0)$, e.g., due to mitigation measures
       and/or behavioral changes (i.e., lockdown phase).
\end{enumerate}
In this setting, we can derive the following version of the law of large
numbers for ages and compartments.

\begin{thm}[One individual] \label{thm:singlePop}
    Suppose that the process is supercritical and that the population is
    started from one individual. Conditional on non-extinction:
    \begin{enumerate}
        \item There exists a r.v.\ $W_\infty$ such that $W_\infty > 0$
            a.s.\ and 
            \[
                \sum_{\sigma_x < t_K} e^{-\alpha t_K}
                \underset{K \to \infty}{\longrightarrow} W_\infty
                \quad \text{in probability.}
            \]
        \item For any $t > 0$, we have 
            \[
                e^{-\alpha t_K} \mu^K_{t_K+t}(\diff a \times \{i\}) 
                \underset{K \to \infty}{\longrightarrow}
                W_\infty n(t,a) p(a,i) \diff a
            \]
            in probability for the weak topology, where $(n(t,a);\, t,a \ge 0)$
            is the solution to \eqref{eq:McKVF} with initial condition
            $g(a) = \alpha e^{-\alpha a}$.
    \end{enumerate}
\end{thm}

This result states that, when the large population size is obtained by
natural population growth, the population has an exponentially
distributed initial age profile with a random size $W_\infty$ determined
by the early infection events. Moreover, the parameter of the exponential
distribution corresponds to the exponential growth rate of the epidemic
prior to the enforcement of control measures. This result can prove
useful in applications, as the exponential growth rate can be readily
estimated from incidence data, whereas the age structure of the
population can hardly be directly assessed. It is a quite generic
phenomenon that the macroscopic behavior of population models started
from a few individuals is described by a deterministic system, with a random
initial condition resulting from the stochasticity of the initial
population growth \cite{barbour16,baker18}.

\paragraph{Summary.} The macroscopic behavior of the epidemic is
characterized by the sole intensity measure $\tau$ and dictates an
explicit age structure of the population. The class structure is deduced
by integrating the life-cycle process against the limiting age profile.
This suggests an alternative point of view on epidemic models, as
age-structured models decorated with classes.

In order to validate our approach, we use those deterministic
approximations to infer epidemiological parameters (reproduction number
before and during lockdown) from recent empirical observations, and show
that our findings are in accordance with the current literature. 

\subsection{Inference on the French COVID-19 epidemic}

We have illustrated the practical interest of our approach by carrying
out parameter inference on data from the early French COVID-19 epidemic.
We focus on two important inference aspects of this epidemic:
providing estimates for key epidemiological quantities, such as the
reproduction number that allows to assess the impact of control measures;
and predicting the number of individuals in ICU and hospital to monitor
the pressure on the healthcare system. 

In our framework, the first task only requires a simple and parsimonious
model that can be adjusted on incidence data, whereas fitting the number
of individuals in ICU requires a more complex model that better accounts
for the population heterogeneity.

\paragraph{The early COVID-19 epidemic in France.}
After a rapid increase in the number of detected cases and deaths, the
French government issued a first nation-wide lockdown from March 17 2020
to May 11 2020. From March 18 2020, it has provided publicly available
daily reports of the number of ICU and hospital admissions, hospital
deaths, as well as the number of occupied ICU and hospital beds, and
discharged individuals. The daily number of detected cases was also
reported, but was considered as unreliable during this period due to the
high variation in the number of tests performed. No additional control
measure was enforced during this period.

\paragraph{Estimating epidemiological parameters from incidence data.} In
order to estimate the impact of lockdown we consider a
parsimonious model that requires to estimate few parameters. It is
illustrated in Figure~\ref{F:admissionModel}, and we refer to it as the
\emph{admission model}. Upon infection, individuals either develop a mild
form of COVID-19 from which they will recover, or a more severe form that
will eventually lead to a hospital admission. Then, hospitalized
individuals either recover and are discharged after some amount of time,
or are moved to ICU. Finally, individuals in ICU either die or recover. A
more detailed description of the model and its parameters is given in
Section~\ref{SS:admission}.

\begin{figure}
    \centering
    \includegraphics[width=.95\textwidth]{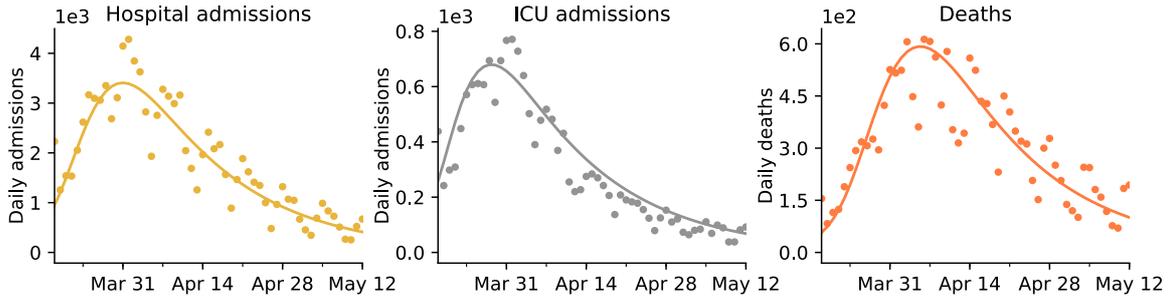}
    \caption{Best fit of the admission model. Solid lines correspond to
    the number of hospital admissions, ICU admissions and deaths
    predicted by the admission model. The dots are the corresponding
    observed values. The dispersion of the observations is mainly due to
    unreported data during the weekends, that are only reported at the
    start of the following week.}
    \label{F:admissionFit}
\end{figure}

We have fitted this model to the following three time series: daily
number of admissions in hospital and ICU, and daily deaths. Fitting such
``incidence'' time series only requires to estimate the entrance time in
each compartment, and not the corresponding sojourn times. The best
fitting model is represented in Figure~\ref{F:admissionFit}, and the
inference procedure is described in Section~\ref{S:inference}. We see
that our simple model reproduces quite well the shapes of the three
incidence time series. The estimated reproduction number after lockdown
is $R_{\mathrm{post}} = 0.745$, and that before lockdown is
$R_\mathrm{pre} = 3.25$. Thus we estimate that the lockdown yielded
a reduction of the reproduction number by a factor $4.36$. Moreover, the
estimated number of infections having occurred before March 18 2020 is 
$W = 9.85 \times 10^5$. All these estimates are in line with that of
other studies on the same dataset \cite{salje_estimating_2020,
sofonea_epidemiological_2020}.

\paragraph{Fitting prevalence data.} Our second objective is to fit three
additional time series: the number of occupied hospital and ICU beds, and
the number of discharged individuals. Using the simple admission model
only yields a poor fit of these new times series, see
Figure~\ref{F:poorFit}. We have identified two main causes for this
discrepancy. First, we have made the simplifying assumption that
individuals are always admitted to ICU prior to death. However, it has
been reported that a large fraction of deaths do not involve a
preliminary ICU admission \cite{lefranck21}, and our assumption leads to
a fraction of deaths among ICU patients much higher than that previously
reported. Second, there are many heterogeneities in the population, such
as the (actual) age, that are known to play an important role in the severity
of the symptoms of COVID-19 and that are not accounted for.

Thus, we used a more detailed model to reproduce all six time series,
which is illustrated in Figure~\ref{F:occupancyModel} and referred to as
the \emph{occupancy model}. Again, a detailed description of the model
and its parameters is available in Section~\ref{SS:occupancy}. The main two differences
with the admission model are that a fraction of individuals die shortly
after hospital admission, and that we distinguish between individuals who
recover fast after their admission and individuals who recover slowly.
The best-fitting model is displayed in Figure~\ref{F:occupancyFit}.
Again, it reproduces quite well the shapes of all the time-series. Under
the occupancy model, the estimated reproduction number after lockdown is
$R_{\mathrm{post}} = 0.734$ and the estimated number of infections before
March 17 2020 is $W = 9.52 \times 10^5$. These estimates are close to
those obtained under the admission model, indicating that the predictions
made by the simple admission model are quite robust to the addition
of model details.

\begin{figure}
    \centering
    \includegraphics[width=.95\textwidth]{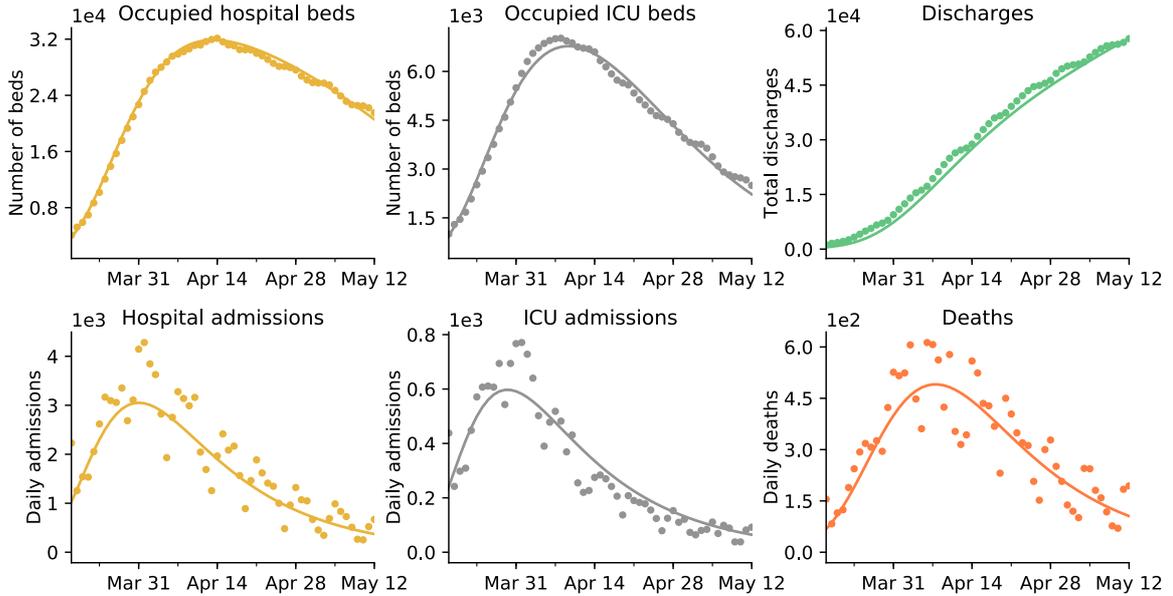}
    \caption{Best fit of the admission model. The solid lines correspond
    to the number of deaths, discharges, occupied ICU and hospital
    beds and ICU and hospital admissions predicted by the occupancy
    model. The dots are the corresponding observed values.}
    \label{F:occupancyFit}
\end{figure}

Overall, our inference work suggests that a simple model can be used to
determine ``global'' epidemiological parameters, such as the reproduction
number and total number of infections, whereas obtaining a prediction for
the number of individuals in hospital or ICU requires to use a more
detailed model that accounts for population heterogeneity. Moreover, it
demonstrates that decorated age-structured models can be readily used to
carry out parameter inference in the context of COVID-19, even when the
underlying compartment structure is quite complex.

Section~\ref{S:inference} contains a detailed description of the
inference procedure, as well as a comparison of the various estimates
that we obtain with estimates currently available in the literature.

\subsection{Connection with ODE models}

Section~\ref{SS:SIR} shows that the SIR model can be seen as a decorated
age-structured model, with well-chosen Markov life-cycle process and
Poisson infection point process. This section extends this representation
to a broader class of ODE models. We will be interested in solutions to
the following set of ODEs:
\begin{align} 
    \begin{split} \label{eq:generalSIR}
    \forall i \in \mathcal{S},\quad 
    \dot{n}_i(t) &= S(t) \sum_{j \in \mathcal{S}} \beta_{ji} n_j(t) + 
    \sum_{j \in \mathcal{S}} q_{ji} n_j(t) \\
    S(t) &= 1 - \sum_{j \in \mathcal{S}} n_j(t).
\end{split}
\end{align}
The parameter $\beta_{ij} \ge 0$ gives the rate of new infections from
individuals in compartment $i$ such that the newly infected individual
starts in compartment $j$. The matrix $T$ with entries $(\beta_{ij})$ is
referred to as the \emph{transmission matrix}. For $i \ne j$, $q_{ij} \ge 0$
corresponds to transition rate from compartment $i$ to compartment $j$.
The \emph{transition matrix} with entries $(q_{ij})$ is denoted by $Q$,
and we further impose that 
\[
    \forall i \in \mathcal{S},\quad q_{ii} = - \sum_{j \ne i} q_{ij}.
\]
This class of ODE models encompasses many common epidemic models,
including the SIR and SEIR models, as well as all models described in
\cite[Chapter~4]{brauer19} for instance.

\begin{prop} \label{prop:ODE}
    \begin{enumerate}
    \item Suppose that $(X(a);\, a \ge 0)$ is a Markov process with jump matrix
    $Q = (q_{ij};\, i,j \in \mathcal{S})$, and that conditional on
    the life-cycle process, $\mathcal{P}$ is a Poisson point process
    with rate $\lambda_i$ when $X(a) = i$. 
    Then, if $(n(t,a);\, t, a \ge 0)$ is a solution to \eqref{eq:Kermack_McK},
    \begin{equation} \label{eq:ODErepr}
        \widetilde{n}_i(t) = \int_0^\infty n(t,a) p(a,i) \diff a
    \end{equation}
    solves \eqref{eq:generalSIR} with $\beta_{ij} = \lambda_i p(0,j)$.
    \item The solution of \eqref{eq:generalSIR} with transmission and
        transition matrices $T$ and $Q$ can be written as
        \eqref{eq:ODErepr} if and only if $\rank(T) = 1$.
    \end{enumerate}
\end{prop}

\begin{proof}
    By differentiating both sides of \eqref{eq:ODErepr} w.r.t.\ time, we
    obtain
    \begin{align*}
        \frac{\diff }{\diff t} \widetilde{n}_i(t)
        &= \int_0^\infty \partial_t n(t,a) p(a,i) \diff a 
        = - \int_0^\infty \partial_a n(t,a) p(a,i) \diff a  \\
        &= n(t, 0)p(0, i) + \int_0^\infty n(t,a) \partial_a p(a,i) \diff a.
    \end{align*}
    By conditioning on the process $X$ and using standard properties of
    Poisson point processes, we can compute the intensity measure of
    $\mathcal{P}$ as  
    \[
        \Angle{\tau, f} = \E\Big[ \int_0^\infty \lambda_{X(a)}f(a) \diff a\Big] 
        = \int_0^\infty f(a) \sum_{i \in \mathcal{S}} p(a,i)\lambda_i \diff a,
    \]
    so that 
    \[
        \forall a \ge 0,\quad \tau(a) = \sum_{i \in \mathcal{S}}
        p(a,i) \lambda_i.
    \] 
    By using the boundary condition of \eqref{eq:Kermack_McK} and the
    fact that $(X(a);\, a \ge 0)$ is a Markov process with generator $Q$
    we get
    \begin{align*}
        \frac{\diff }{\diff t} \widetilde{n}_i(t)
        &= p(0,i) S(t) \int_0^\infty n(t,a) \sum_{j \in \mathcal{S}} \lambda_j p(a,j) \diff a 
        + \int_0^\infty n(t,a) \sum_{j \in \mathcal{S}} q_{ji} p(a,j) \diff a \\
        &= S(t) \sum_{j \in \mathcal{S}} p(0,i) \lambda_j \widetilde{n}_j(t)
        + \sum_{j \in \mathcal{S}} q_{ji} \widetilde{n}_j(t)
    \end{align*}
    from which the first item follows.

    It is clear that if $\beta_{ij} = \lambda_i p(0,j)$ then $\rank(T) = 1$.
    Fix some $T$ and $Q$. If $\rank(T) = 1$, then it can be decomposed as 
    \[
        \forall i,j \in \mathcal{S},\quad \beta_{ij} = \lambda_i p_j 
    \]
    where $(p_j;\, j \in \mathcal{S})$ is a probability vector and 
    $\lambda_i \ge 0$. This decomposition can for instance be recovered from
    \[
        \lambda_i = \sum_{j \in \mathcal{S}} \beta_{ij}, \quad p_j =
        \frac{\beta_{ij}}{\lambda_i}.
    \]
    Consider a Markov process $(X(a);\, a \ge 0)$ with transition matrix $Q$
    and $X(0)$ distributed as $(p_j;\, j \in \mathcal{S})$. If $\mathcal{P}$
    is such that infections occur at rate $\lambda_i$ when $X(a) = i$, 
    the first item proves that \eqref{eq:generalSIR} can be written as
    \eqref{eq:ODErepr}.
\end{proof}

The previous result provides a simple criterion for a system of ODEs to
be represented as a decorated age-structured PDE. This criterion has
already been proposed previously and is referred to as the ``separable
mixing'' assumption \cite{Diekmann1990, Diekmann2013}. A direct
consequence of this result is that not all ODE models can be represented
using a single decorated age-structured PDE. However many models fulfill
the requirement that $\rank(T) = 1$, including models with a single
infectious state, those where new infected individuals always start in
the same state, and all classical models exposed in
\cite[Chapter~4]{brauer19}. In that sense, our framework greatly extends
the usual systems of ODEs widely used in epidemic modeling.

An important situation where $\rank(T) > 1$ is that of heterogeneous
contact rates in the population. For instance, the contact rate could
depend heavily on the (actual) age group to which individuals belong,
and more contacts are made within the same age class than between
classes. A second example is that of spatial heterogeneity, where
contacts are more likely to occur between spatially close individuals.
It remains possible to derive a representation of \eqref{eq:generalSIR}
in the case $\rank(T) > 1$ using an age-structured model, but this
would require to use a multi-dimensional version of \eqref{eq:Kermack_McK}, 
which is a straightforward extension of this model.

\subsection{Relation with previous works and outline}

Deterministic epidemic models where the infectivity depends on the
individuals' age of infection were first introduced in
\cite{kermack1927}, and their mathematical properties have been further
studied thoroughly \cite{reddingius1971, diekmann1977, thieme1985}, see
\cite{inaba17} for a recent account. However, these models have received
surprisingly little attention in applications compared to their ODE
counterpart, which have been widely used for instance in the context of
COVID-19 \cite{roques_using_2020, salje_estimating_2020,
evgeniou_epidemic_2020, DiDomenico2020, djidjou_optimal_2020}. In this
direction, let us mention \cite{forien2021} which makes use of an
epidemic model with memory and \cite{gaubert20} were a set of
transport PDEs similar to \eqref{eq:Kermack_McK} is used. Note that,
in contrast with our approach, the PDE in the latter work has two
dimensions and is structured according to the time since the entrance in
the compartments rather than by infection age. We hope that our work
illustrates well the practical potential of such general models. The
relation between age-structured and ODE epidemic models exposed in
Section~\ref{SS:SIR} is known since their very introduction
\cite{kermack1927}, and has been acknowledged multiple times since then
\cite{metz1978, diekmann1995, brauer05, inaba17}. 

In the most general formulation of a CMJ process, individuals can carry a
trait valued in an abstract measure space that encodes all the information
about their infection \cite{jagers1975branching, Taib1992}. We have
restricted this information to the sequence of compartments visited by
each individual, but we could have included some additional details, such
as the evolution of the viral load for instance, which could be modeled
as a continuous trait following a diffusion. Our result would carry over
to the general setting, with a modified limiting equation in
Theorem~\ref{thm:largePop} which has already been proposed in \cite[Eq.\
(2.5)]{metz1978} and \cite[Chapter~IV, Section~1.3]{Metz1986}. (Note that
none of these works is concerned with the underlying stochastic model.)
However, we believe that our current formalism, where the state of an
infected individual can be described by a discrete set of compartments,
is flexible enough for applications, while being easier to grasp than the
general case.

Non-Markov epidemic models have already been investigated, see
e.g.~\cite{sellke1983asymptotic, pang2020functional, ball1986unified,
barbour1975duration}. In particular our work shares similarities with a
recent series of work in this direction \cite{pang2020functional,
forien2021epidemic, pang2021functional}. Let us briefly show how to
translate the model in \cite{forien2021epidemic} in our framework. Let
$\lambda = (\lambda(a);\, a \ge 0)$ be some random function giving the
infectiousness of a typical individual in the population. Conditional on
$\lambda$, let $\mathcal{P}$ be a time-inhomogeneous Poisson point
process with intensity $\lambda$, define 
\[
    \zeta = \inf \{ a : \lambda(a) > 0 \},\qquad \eta + \zeta = \sup \{ a : \lambda(a) > 0 \}
\]
and set 
\[
    \forall a \ge 0,\quad X(a) =
    \begin{cases}
        E &\text{if $a \in [0, \zeta)$} \\
        I &\text{if $a \in [\zeta, \zeta + \eta)$} \\
        R &\text{if $a \in [\zeta + \eta, \infty)$}.
    \end{cases}
\]
Then, up to the choice of the initial condition and the fact that we make
a branching assumption, the model considered in \cite{forien2021epidemic}
coincides with the model considered here for this specific choice of the
pair $(\mathcal{P}, X)$. (It coincides exactly with the extension of our
model considered in \cite{duchamps2021general}.)

The main result in \cite{forien2021epidemic} shows that the fraction of
individuals in each state ($S$, $E$, $I$, and $R$) converges to the
solution of a set of Volterra equations, see their Theorem~2.1. A
straightforward computation shows that the latter equations are
equivalent to the non-linear version our decorated
McKendrick-von~Foerster PDE obtained by replacing \eqref{eq:McKVF} by
\eqref{eq:Kermack_McK}, with the specific choice of $(\mathcal{P}, X)$
described above.

The idea of representing a general branching
population by its age structure has a rich history in probability theory
\cite{jagers1975branching, fan2020convergence, Jagers1984,
jagers2011population, jagers2000population, hamza2013age, tran2008large,
ferriere2009stochastic} and the connection with the McKendrick-von
Foerster PDE has been acknowledged several times \cite{fan2020convergence, 
hamza2013age}. In the latter two works, the authors allow for birth and
death rates that may depend not only on abundances of each type, but also
on the whole age structure of the population. This impressive level of
generalization comes at the cost of assuming that the process describing
the evolution of the empirical measure on ages and types is Markovian. In
particular, birth and death rates are not allowed to depend on past
individual birth events. The Markov property then allows the use of a
generator for the empirical measure and with some extra finite second
moment assumptions on the intensity measure, this approach allows the
authors to obtain a law of large numbers and a central limit theorem.  

Even if the current work is not as mathematically challenging as that
alluded to above, we believe that our point of view does deserve to be
highlighted in the current sanitary crisis since it provides both a
general modeling framework and an efficient inference methodology.
Furthermore, since we ignore finite population effects, our proofs are
quite elementary compared to \cite{fan2020convergence, hamza2013age} and
should be accessible to a much wider audience interested in such a
modeling approach. Finally, as far as we can tell, the duality result
exposed in Section~\ref{SS:dualFK} is new and can presumably be extended
to more general branching processes where birth and death rates are
allowed to be frequency-dependent. In \cite{duchamps2021general}, some of the authors
of the present work show that this duality result has a natural
counterpart in a model with a finite but large population.

\paragraph{Outline.}
The remainder of the paper is organized as follows.
Section~\ref{S:feynmanKac} is devoted to the study of Eq.\
\eqref{eq:McKVF}. After providing a formal construction of the branching
process that we consider in Section~\ref{SS:assumption}, the definition
of a weak solution to \eqref{eq:McKVF} is given in
Section~\ref{SS:weakSolution}. Then, we derive two probabilistic
representations of this solution: we show in
Section~\ref{SS:forwardFormula} that it corresponds to the first moment
of the branching process that we are studying, when viewed as a random
measure on the ages of infection; Section~\ref{SS:dualFK} provides a
construction of the weak solution using a dual genealogical process. The
two laws of large numbers are proved in Section~\ref{S:proofs}. Finally,
Section~\ref{S:inference} describes the inference procedure, and compares
the estimates that we obtain to known estimates from the literature.

\section{Two Feynman-Kac formul\ae}
\label{S:feynmanKac}

\subsection{Assumptions and notation}
\label{SS:assumption}

\paragraph{CMJ branching process with suppression.}
Recall that the infection process is modeled by a Crump-Mode-Jagers (CMJ)
branching process \cite{jagers1975branching, Nerman1981} with no death,
starting from one individual called the progenitor (or root of the tree).
It can be briefly constructed as follows.

Using the Ulam-Harris labeling, the population can be indexed by
\[
    \mathscr{U} \coloneqq \{\emptyset\} \cup \bigcup_{n \ge 1}^\infty
    \mathbb{N}^n.
\]
The set $\mathscr{U}$ encodes a tree where $xi \coloneqq (x,i)$ is the
$i$-th child of $x$. Each individual $x \in \mathscr{U}$ is characterized by a
pair $(\mathcal{P}_x, X_x)$ embodying respectively the processes of
secondary infection events from $x$ and of types carried by $x$. Each
pair $(\mathcal{P}_x, X_x)$ is an i.i.d.\ copy of the pair $(\mathcal{P},
X)$ with law $\mathscr{L}$, except when $x$ is the root, where it is
distributed as $(\tilde{\cal P}, \tilde{X})$ with law
$\tilde{\mathscr{L}}$ (more on that below). 

An infection time $\sigma_x$ can be assigned to all individuals
inductively as follows, with the convention that $\sigma_x = \infty$ for
individuals that are not infected. Suppose that $\sigma_\emptyset$ is
known (see below). Then, if $\sigma_x < \infty$ has been defined, let
$A_1, \dots, A_{N_x}$ denote the atoms of $\mathcal{P}_x$ in increasing
order. That is,
\[
    \mathcal{P}_x = \sum_{i = 1}^{N_x} \delta_{A_i}
\]
with $A_1 < \dots < A_{N_x}$. Set $\sigma_{xi} = \infty$ for $i > N_x$,
and, independently for each $i \le N_x$, set
\[
    \sigma_{xi} = \begin{cases}
        \sigma_x + A_i &\text{ with probability $c(\sigma_x + A_i)$}\\
        \infty &\text{ with probability $1-c(\sigma_x + A_i)$},
    \end{cases}
\]
where we recall that $(c(t);\, t \ge 0)$ is the contact rate. 
This amounts to trimming the tree by pruning the subtree stemming 
from $x$ with probability $1-c(\sigma_x)$, see Figure~\ref{F:CMJ}.

\begin{figure}
\centering
\includegraphics[scale=0.5]{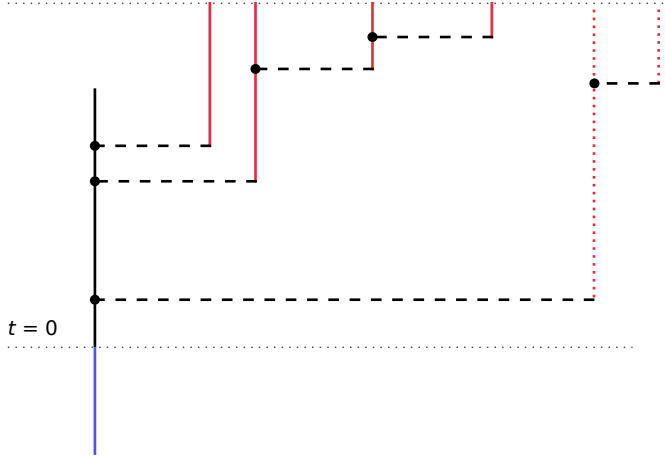}
\caption{The initial individual  $(\tilde {\cal P}, \tilde X)$ is
    represented by a black segment. In Section~\ref{SS:assumption}, we
    assume that at time $t=0$, the age of the initial individual (length
    of the blue segment) is distributed according to a probability
    density $g$.  If a branching event is observed at time $t$ (see e.g.,
    black dots), the infection occurs with probability $c(t)$. In the
    CMJ, this amounts to prune the corresponding subtree with probability
    $c(t)$ (dotted red tree).}  
    \label{F:CMJ}
\end{figure}

\paragraph{Initial shifted law.} In order to connect the distribution of
the CMJ to the McKendrick-von~Foerster equation, we allow the progenitor
to have an initial age with an arbitrary distribution. Let $A$ be a r.v.\
distributed according to some density $g$. Define the infection time of
$\emptyset$ as $\sigma_\emptyset = -A$. The secondary infections induced
by the progenitor occur at some times $\sigma_\emptyset +
\widetilde{A}_1$, \dots, $\sigma_{\emptyset} + \widetilde{A}_{\widetilde{N}}$, 
where $(\widetilde{A}_1, \dots, \widetilde{A}_{\widetilde{N}})$ are the
atoms of a point process $\widetilde{\mathcal{P}}$ defined as 
\[
    \widetilde{\mathcal{P}} = \sum_{A_i \in \mathcal{P}}
   \mathbbm{1}_{\{A_i > A\}} \delta_{A_i},
\]
where the pair $(\mathcal{P}, X)$ has law $\mathscr{L}$. The point
process $\widetilde{\mathcal{P}}$ is obtained from $\mathcal{P}$ by
erasing all atoms that would lead to an infection before $t=0$. Define
$\widetilde{X} = X_\emptyset$, and let $\tilde{\mathscr{L}}$ be the
distribution of $(\widetilde{\mathcal{P}}, \widetilde{X})$. We refer to
$\tilde{\mathscr{L}}$ as the initial shifted law. The infection times
$(\sigma_x;\, x \in \mathscr{U} \setminus \{\emptyset\})$ are then
defined recursively as above, from i.i.d.\ pairs $(\mathcal{P}_x, X_x;\,
x \in \mathscr{U}\setminus \{\emptyset\})$ with the original law
$\mathscr{L}$.

\paragraph{Assumptions.} The following assumptions will be made
implicitly in the remainder of our work. For simplicity, we assume that
the contact rate $(c(t);\, t \ge 0)$ is a piecewise continuous function,
and that for any $a \ge 0$, the process $(X(a);\, a \ge 0)$ is a.s.\
continuous at $a$.

Recall that the intensity measure of the point process $\mathcal{P}$ 
is denoted by $\tau$, and implicitly defined as 
\[
    \Angle{\tau, f} = \E \big[ \Angle{\mathcal{P}, f} \big]
\]
for any test function $f$, where we used the notation $\Angle{\mu, f} =
\int f \diff \mu$. We assume that $\tau$ has a density w.r.t.\ the
Lebesgue measure that we still denote by $(\tau(a);\, a \ge 0)$, and
assume that
\[
    R_0 \coloneqq \int_0^\infty \tau(a) \diff a < \infty.
\]
We also assume that there exists a unique parameter $\alpha \in \R$,
the so-called Malthusian parameter of the (untrimmed) CMJ process, such
that
\begin{equation} \label{eq:malthusian}
    \int_0^\infty \exp(-\alpha a) \tau(a) \diff a = 1. 
\end{equation}
The parameter $\alpha$ can be either positive (supercritical case) or
negative (subcritical case).

\subsection{McKendrick-von Foerster PDE: Weak solutions}
\label{SS:weakSolution}

This section provides the existence and uniqueness of weak solutions to
Eq.~\eqref{eq:McKVF}. Even if similar results are well-known for the
time-homogeneous McKendrick-von~Foerster equation \cite[Chapter~1]{inaba17}, 
we derive them briefly here for the sake of completeness.

In order to motivate our definition of weak solutions, we start by giving
a well-known formal resolution of the PDE using the method of
characteristics. Fix $a>0$. Let
\[
    A(t) = a - t 
\] 
Then
\begin{align*}
    \frac{\diff}{\diff s} n(t-s, A(s)) = 
    - \partial_t n(t-s,A(s)) - \partial_a n(t-s,A(s))=0,
\end{align*}
so that $s\mapsto n(t-s,a-s)$ is conserved along the characteristics, i.e.,
\[
    \forall s<a, \ \ n(t,a) = n(t-s,a-s).
\]
It follows that
\begin{align}\label{eq:n}
    n(t,a) = 
    \begin{cases}
        g(a-t) &\text{when $a > t$}\\
        b(t-a) &\text{when $a\leq t$}
    \end{cases}
\end{align}
where 
\[
    b(t) =  c(t) \int_0^\infty n(t,a) \tau(a)\diff a 
\]
is the number of new infections at time $t$. We now determine the
function $b$. Injecting the previous expression into the ``age'' boundary
condition of the PDE, we obtain a convolution equation for $b$: 
for every $t>0$
\begin{equation} \label{eq:renewal}
    b(t) = c(t) \int_0^t b(t-a) \tau(a)\diff a  +  c(t)
    \int_t^\infty g(a-t) \tau(a) \diff a.
\end{equation}
Recall that $\alpha$ denotes the Malthusian parameter defined in
\eqref{eq:malthusian}.

\begin{lem}
    There exists a unique solution $b$ to~\eqref{eq:renewal} which is
    locally integrable. Moreover, for any $\delta \ge 0$ such that
    $\delta >\alpha$ we have $b \in \mathscr{L}^{1,\delta}$, where 
    $\mathscr{L}^{1,\delta}$ denotes the set of all functions 
    $f\colon \mathbb{R}_+\to \mathbb{R}$ such that 
    $\norm{f}_{L^{1,\delta}} \coloneqq \int_0^\infty e^{-\delta t}
    \abs{f(t)} \diff t < \infty$.
\end{lem}

\begin{proof}
Fix $\delta > \alpha$ and let $L^{1,\delta}$ denote the quotient space 
of $\mathscr{L}^{1,\delta}$ by the relation $\sim_\delta$, where
$f\sim_\delta g$ if $\norm{f-g}_{L^{1,\delta}}=0$. Then define the linear
operator $\Phi\colon L^{1,\delta}\to
L^{1,\delta}$ by
\[
    \Phi f\colon t \mapsto c(t) \int_0^t f(t-u) \tau(u) \diff u.
\]
Then we have
\begin{align*}
    \norm{\Phi f}_{L^{1,\delta}} &= \int_0^\infty e^{-\delta t} \Phi f(t) \diff t 
    = \int_0^\infty e^{-\delta t} c(t) \int_0^t f(t-u) \tau(u) \diff u
    \diff t\\
    &= \int_0^\infty e^{-\delta u} f(u) \int_u^\infty
    \tau(t-u)e^{-\delta (t-u)} c(t) \diff t \diff u.
\end{align*}
Now using that 
\[
    \int_u^\infty \tau(t-u) e^{-\delta (t-u)} c(t) \diff t \le \int_0^\infty
    \tau(t) e^{-\delta t} \diff t < 1
\]
we obtain that $\norm{\Phi} < 1$. Define
\[
    \Psi \coloneqq \Id - \Phi.
\]
Then $\Psi$ is invertible with inverse $\sum_{k \ge 0} \Phi^k$.
Note that equation~\eqref{eq:renewal} can be written as
\[
    \Psi(b) = F,
\]
where
\[
    F\colon t \mapsto c(t) \int_t^\infty \tau(a) g(a-t) \diff a.
\]
Noting that $F \in L^{1,\delta}$ as 
\[
    \int_0^\infty e^{-\delta t} F(t) \diff t 
    \le \int_0^\infty \int_t^\infty \tau(a) g(a-t) \diff a \diff t <
    \infty
\]
proves existence and uniqueness of the solution $b$ to \eqref{eq:renewal}
in $L^{1,\delta}$. Now for any two functions $b_1$ and $b_2$ such that
$b_1\sim_\delta b_2$ and $b_1$ and $b_2$ both satisfy \eqref{eq:renewal},
we have $b_1=b_2$ (i.e., there is a single element in the equivalence
class of $b$ verifying \eqref{eq:renewal} for all $t$). This shows
uniqueness of the solution $b$ to \eqref{eq:renewal} in 
$\mathscr{L}^{1,\delta}$.

Since all elements of $\mathscr{L}^{1,\delta}$ are locally integrable,
this also shows the existence of a locally integrable solution to
\eqref{eq:renewal}. Its uniqueness can be proved following the exact same
reasoning as previously, replacing integrations on $[0,\infty)$ by
integration on compact intervals.
\end{proof}

\begin{defn}
    We say that $(n(t,a);\, t,a\geq0)$ is the weak solution to the
    McKendrick-von~Foerster PDE with initial condition $g$ if it
    satisfies the relation \eqref{eq:n} where $(b(t);\, t\geq0)$ is the
    unique locally integrable solution to \eqref{eq:renewal} displayed in
    the previous lemma. 
\end{defn}

\subsection{A forward Feynman-Kac formula}
\label{SS:forwardFormula}

Consider a CMJ with initial shifted law and define
\[
    Z(t) \coloneqq \sum_{x} \mathbbm{1}_{\{\sigma_x\in(0,t]\}}, 
    \quad  B(t) \coloneqq \E\big( Z(t) \big) 
\]
where $Z(t)$ is interpreted as the number of infections between $0$ and $t$. 
Recall that $R_0 = \int_0^\infty \tau(u) \diff u<\infty$ guarantees
that $B(t)<\infty$ for all $t\geq0$. Finally, $B$ is non-decreasing and
we denote by $\diff B$ the Stieljes measure associated to $B$.

\begin{lem}\label{lem:mcvonf}
    There exists a locally integrable function $(b(t);\, t\geq0)$ such that 
    \[
        \diff B(t) = b(t) \diff t.
    \]
    Further, $b$ coincides with the unique locally integrable solution of
    the convolution equation~\eqref{eq:renewal}. 
\end{lem}

\begin{proof}
The fact that $\diff B$ has a density easily follows from the fact that $\tau$
has a density. The fact that $B(t) < \infty$ ensures that $b$ is locally
integrable.

Define $\bar{\mathcal{P}}_x$ the infection measure obtained from
$\mathcal{P}_x$ after random thinning by the function $(c(t);\, t\geq0)$.
Namely, conditional on $\sigma_x$ and the atoms $A_1<A_2<\cdots$ of
$\mathcal{P}_x$, we remove independently each of the atoms with respective
probabilities $1-c(\sigma_x+A_1), 1-c(\sigma_x+A_2), \dots$, whereas the
other atoms remain unchanged.

Fix $t>0$. Let $k\leq n\in{\mathbb N}$. Define 
${\mathbb T}^{k,n}(\mathcal{P}_x)$ as the measure obtained from 
$\mathcal{P}_x$ as follows. Conditional on the atoms $A_1<A_2<\cdots$ of
$\mathcal{P}_x$, we remove independently each of the atoms with
respective probabilities 
\[
    1-\max_{z\in(t\frac{k}{n},t\frac{k+1}{n}]} c(z+A_1),
    1-\max_{z\in(t\frac{k}{n},t\frac{k+1}{n}]} c(z+A_2), \cdots
\] 
and leave other atoms unchanged. Note that the thinning procedure is now
independent of the starting time $\sigma_x$. Further, if
$\sigma_x\in(t\frac{k}{n},t\frac{k+1}{n}]$, the point measure ${\mathbb
T}^{k,n}(\mathcal{P}_x)$ dominates $\bar{\mathcal{P}}_x$.

We decompose the births on $(0,t]$ into two parts:
individuals stemming from the root $\emptyset$ and a second part from
subsequent births. Using the fact that for every
individual $x$, the (un-suppressed) random measure $\mathcal{P}_x$ is
independent  of its birth time $\sigma_x$ (see second equality below), and setting 
$M(t) \coloneqq \int_0^t \int_0^\infty g(a) \tau(a+u) c(u) \diff a \diff u$,
we get
\begin{align*} 
    B(t) &= \sum_{k=0}^{n-1} \sum_{x \ne \emptyset} 
            \mathbb{E} \bigg(  
                \mathbbm{1}\Big( \sigma_x \in \Big(t \frac{k}{n}, t\frac{k+1}{n}\Big] \Big) 
                \int_0^{t-\sigma_x} \bar {\cal P}_x (\diff a) 
            \bigg) + M(t) \\ 
         &\leq \sum_{k=0}^{n-1} \sum_{x \ne \emptyset} 
            \mathbb{E} \bigg( 
                \mathbbm{1} \Big(\sigma_x \in \Big(t \frac{k}{n}, t\frac{k+1}{n}\Big] \Big) 
                \int_0^{t-t\frac{k}{n}} \mathbb{T}^{k,n}(\mathcal{P}_x)(\diff a)  
            \bigg) + M(t) \\
         & = \sum_{k=0}^{n-1} \sum_{x \ne \emptyset} 
             \mathbb{E} \bigg(  
                 \mathbbm{1}\Big(\sigma_x \in \Big(t \frac{k}{n}, t\frac{k+1}{n}\Big] \Big) 
             \bigg) 
             \mathbb{E} \bigg(
                 \int_{0}^{t-t\frac{k}{n}} \mathbb{T}^{k,n}(\mathcal{P})(\diff a)
             \bigg) + M(t) \\
         & = \sum_{k=0}^{n-1} 
             \mathbb{E} \bigg( \sum_{x \ne \emptyset} 
                 \mathbbm{1}\Big(\sigma_x \in \Big(t \frac{k}{n}, t\frac{k+1}{n}\Big]\Big) 
             \bigg) 
             \mathbb{E} \bigg(
                 \int_{0}^{t-t\frac{k}{n}} \mathbb{T}^{k,n}(\mathcal{P})(\diff a)
             \bigg) + M(t) \\
         & = \sum_{k=0}^{n-1} 
             \bigg( B\Big(t\frac{k+1}{n}\Big) - B\Big(t\frac{k}{n}\Big) \bigg) 
             \mathbb{E} \bigg(
                 \int_{0}^{t-t\frac{k}{n}} \mathbb{T}^{k,n}(\mathcal{P})(\diff a)
             \bigg) + M(t)\\
         & = \sum_{k=0}^{n-1} 
             \bigg( B\Big(t\frac{k+1}{n}\Big) -  B\Big(t\frac{k}{n}\Big) \bigg) 
             \int_{0}^{t-t\frac{k}{n}} c^{k,n}(u)\tau(u) \diff u + M(t).
\end{align*}
with $c^{k,n}(y) = \max_{v\in(t \frac{k}{n},t\frac{k+1}{n}]} c(y+v)$. In
particular, if $t k/n\to x$, and $x+y$ is a continuity point of $c$, we have
$c^{k,n}(y)\to c(x+y)$. We will pass to the limit $n\to\infty$ in the
latter inequality. Recall that $c$ is bounded (and valued in $[0,1]$) and
right-continuous. The first term on the RHS can be written
under the form
\[ 
    \sum_{k=0}^{n-1} \bigg( B\Big(t\frac{k+1}{n}\Big) - B\Big(t\frac{k}{n}\Big) \bigg)
        \int_0^{t-t\frac{k}{n}} c^{k,n}(u)\tau(u) \diff u  
    = \int_0^t f^{(n)}(y) \diff B(y),
\]
where 
\[
    f^{(n)}(y) = \int_0^{t-[y]_n} 
    \tau(u) \sup_{v \in ([y]_n,\; [y]_n + \frac{t}{n}]} c(v+u)  \diff u 
    \quad\text{and} \quad [y]_n  = \frac{t}{n} \floor{ny/t}.
\]
We will now apply twice the Bounded Convergence Theorem. On the one hand,
for a fixed value of $y$, as $n \to \infty$
\[
    \mathbbm{1}_{[0, t-[y]_n]}(u) \tau(u) \sup_{v \in ([y]_n, [y]_n +\frac{t}{n}]} c(v+u) 
    \longrightarrow \mathbbm{1}_{[0, t-y]}(u) \tau(u) c(y+u) \quad \text{Lebesgue a.e.} 
\]
Further, the latter term (i.e., the integrand in the integral defining
$f^{(n)}$) is uniformly bounded by $\tau$ and $\int_0^\infty \tau(u) \diff u < \infty$. 
A first application of the Bounded Convergence Theorem implies that for
every $y$, as $n \to \infty$
\[
    f^{(n)}(y) \to \int_0^{t-y} c(y+u) \tau (u) \diff u.
\]
On the other hand, the uniform bound, $f^{(n)}(y)\leq R_0 = \int_0^\infty \tau(u) \diff u$ 
for all $y,n$, allows us to again apply the Bounded Convergence Theorem,
so we get
\begin{align*} 
    B(t) \leq  \int_{0}^t b(y) \int_{0}^{t-y} c(y+u) \tau(u) \diff u \diff y 
             + \int_0^t \int_0^\infty g(a) \tau(a+u) c(u) \diff a \diff u. 
\end{align*}
By replacing the $\max$ by a $\min$ and using a similar argument, one can
establish the same lower bound and strengthen the latter inequality into
an equality. A simple change of variable $v=u+y$ and interchanging the
order of integration yields
\[
    B(t) = \int_{0}^t c(v)  \int_{0}^{v} \tau(v-y) b(y) \diff y \diff v
    + \int_0^t \int_0^\infty g(a) \tau(a+u) c(u) \diff a \diff u.  
\] 
Finally, differentiating with respect to $t$ yields the desired result.
\end{proof}

\begin{cor}[Forward Feynman-Kac formula] \label{cor:forwardFK}
    For every $t \ge 0$, define
    \[
        \bar \mu_t(\diff a\times \{i\}) \coloneqq  n(t,a) \times
        \mathbb{P}(X(a)=i) \diff a , 
    \]
    where $n$ is the unique weak solution to the McKendrick-von~Foerster
    PDE with initial condition $g$. Then
    \begin{equation}
        \bar \mu_t(\diff a\times \{i\}) 
            = \mathbb{E}\Big(\sum_{x} \mathbbm{1}_{\{\sigma_x<t\}} 
            \delta_{(t-\sigma_x, X_x(t-\sigma_x))}(\diff a\times \{i\})\Big) 
    \end{equation}
    where the expected value is taken with respect to a CMJ process
    starting with one individual with infection and life-process
    distributed according to the shifted law $\tilde{\mathscr{L}}_g$. 
\end{cor}

\begin{proof}
Define
\[
    \bar \mu'_t(\diff a\times \{i\}) \coloneqq 
    \mathbb{E}\Big(\sum_{x} \mathbbm{1}_{\{\sigma_x<t\}} \delta_{\left(t-\sigma_x,
        X_x(t-\sigma_x)\right)}(\diff a\times \{i\}) \Big)  
\]
We need to check that $\bar \mu'_t=\bar \mu_t$ on the space of finite
measures. Let $F$ be a non-negative, bounded, continuous
function on $\mathbb{R}_+ \times \mathcal{S}$ and $h$ a non-negative,
continuous function with compact support in $\R_+$. As in
the previous lemma, we have
\begin{align*}
    \int_0^\infty  h(t) \int F(a,i) \bar \mu'_t(\diff a, \diff i) \diff t 
        &=  \sum_{x\neq \emptyset} 
        \mathbb{E}\Big(\int_0^\infty h(t)F\big(t-\sigma_x, X_x(t-\sigma_x)\big)
                \mathbbm{1}_{\{\sigma_x<t\}} \diff t  \Big) \\ 
        &+ \int_0^\infty \int_0^\infty h(t) \mathbb{E}\big(F(t+a, X(t+a)\big) g(a) \diff a \diff t.
\end{align*}
Let $(I)$ be the first term on the RHS. For every $n\in{\mathbb N}^*$
\begin{align*}
    (I) & = \sum_{k \geq 0} \sum_{x \ne \emptyset}
        \mathbb{E}\bigg( \int_{\sigma_x}^\infty
             h(t)F\big(t-\sigma_x, X_x(t-\sigma_x)\big) 
             \mathbbm{1}\Big(\sigma_x \in \Big(\frac{k}{n}, \frac{k+1}{n}\Big]\Big)
        \diff t \bigg) \\
        & = \sum_{k \geq 0} \sum_{x \ne \emptyset}
        \mathbb{E}\bigg( \int_0^\infty
            h(t+\sigma_x) F\big(t, X_x(t)\big) 
             \mathbbm{1}\Big(\sigma_x \in \Big(\frac{k}{n}, \frac{k+1}{n}\Big]\Big)
        \diff t \bigg) \\
        & \leq \sum_{k \geq 0} \sum_{x \ne \emptyset} 
        \mathbb{E}\bigg( \int_0^\infty 
            \max_{u \in (\frac{k}{n}, \frac{k+1}{n}]} h(t+u) F\big(t, X_x(t)\big) 
            \mathbbm{1}\Big(\sigma_x \in \Big(\frac{k}{n}, \frac{k+1}{n}\Big]\Big) 
        \diff t \bigg) \\
        & = \sum_{k \geq 0} \sum_{x \ne \emptyset} 
        \int_{0}^\infty \max_{u \in (\frac{k}{n}, \frac{k+1}{n}]} h(t+u)
            \mathbb{E} \bigg( F\big(t, X(t)\big) \bigg) 
            \mathbb{P} \bigg(\sigma_x\in \Big(\frac{k}{n}, \frac{k+1}{n} \Big]\bigg) 
            \diff t \\
        & = \sum_{k \geq 0} \int_{0}^\infty \max_{u \in (\frac{k}{n}, \frac{k+1}{n}]} 
            h(t+u) 
            \mathbb{E}\bigg( F\big(t, X(t) \big) \bigg)
            \bigg(B\Big(\frac{k+1}{n}\Big) - B\Big(\frac{k}{n}\Big) \bigg) \diff t.
\end{align*} 
By reasoning along the same lines as in Lemma~\ref{lem:mcvonf} (i.e.,
applying the Bounded Convergence Theorem several times), one can
show that the RHS converges to 
\[
    \int_0^\infty \int_0^\infty h(t+y) \E \Big(F\big(t, X(t)\big)\Big) b(y) \diff t \diff y 
\]
as $n\to\infty$ and thus
\begin{align*}
    \int_0^\infty h(t) \int F(a,i) \bar \mu'_t(\diff a, \diff i) \diff t 
        &\leq \int_0^\infty \int_0^\infty  h(t+y) \E \Big( F\big(t, X(t)\big)\Big) 
            b(y) \diff t \diff y \\
        &\qquad + \int_0^\infty h(t) \int_0^\infty \mathbb{E}\Big(F\big(t+a, X(t+a)\big)\Big) 
            g(a) \diff a \diff t.
\end{align*}
By a similar argument, the inequality can be strengthened into an equality.
After some simple changes of variables we get
\begin{align*}
    \int_0^\infty h(t) \int F(a,i) \bar \mu'_s(\diff a, \diff i) \diff t 
        &= \int_0^\infty h(t) \int_0^t \E \Big(  F\big(a, X(a)\big)\Big) 
            b(t-a) \diff a \diff t \\
        &\qquad + \int_0^\infty h(t) \int_0^\infty \mathbb{E}\Big(F\big(t+a, X(t+a)\big)\Big) 
            g(a) \diff a \diff t.
\end{align*}
Moreover we have
\begin{align*}
    \int_0^\infty h(t)\int F(a,i)\bar{\mu}_t(\diff a, \diff i) \diff t 
    &= \sum_{i \in \mathcal{S}} \int_0^\infty h(t) 
                \int_0^\infty F(a,i) n(t,a) p(a,i) \diff a \diff t \\
    &= \begin{multlined}[t] \sum_{i \in \mathcal{S}} \bigg[ 
            \int_0^\infty h(t) \int_0^t F(a,i) b(t-a) p(a,i) \diff a \diff t \\
            \qquad\qquad 
            + \int_0^\infty h(t) \int_t^\infty F(a,i) g(a-t) p(a,i) \diff a \diff t \bigg]
    \end{multlined}
\end{align*}
so that 
\[
    \int_0^\infty h(t) \int F(a,i) \bar{\mu}_t(\diff a, \diff i) \diff t
    = \int_0^\infty h(t) \int F(a,i) \bar{\mu}'_t(\diff a, \diff i) \diff t.
\]
It is easy to check that the two functions $t \mapsto \Angle{\bar \mu_t,F}$ 
and $t \mapsto \Angle{\bar \mu'_t, F}$ are both continuous. As a
consequence, we have $\Angle{\bar \mu_t,F} = \Angle{\bar \mu'_t, F}$ for
every test function $F$, concluding the proof.
\end{proof}

\subsection{Dual CMJ process and backward Feynman-Kac formula}
\label{SS:dualFK}

We end this section by making a connection between a dual process --
interpreted as an \emph{ancestral process} -- and the (PDE) method of
characteristics. In addition, this approach provides a probabilistic
proof of uniqueness for the PDE.

Let $\mathcal{M}$ be any random point measure with intensity measure
$\tau(\diff u)$. Fix $a,T>0$. We now construct a dual process using the
measure $\mathcal{M}$, which can be seen as a generalized Bellman-Harris
branching process (individuals have a finite lifetimes, births only occur
upon death). Let us first describe the process with no suppression (i.e.,
$c \equiv 1$).
\begin{itemize}
    \item Start with a single particle at time $t=0$. Assume that the
        residual  lifetime of this original particle is $a$, so that this
        particle dies out at time $a$.
    \item As in a Bellman-Harris process, the number of offspring of an
        individual and their lifetime durations are independent of the
        parent's characteristics.
    \item Upon death, each individual $x$ is endowed with an independent
        copy $\mathcal{M}_x$ of $\mathcal{M}$: the number of offspring of $x$
        is given by the number of atoms of $\mathcal{M}_x$ and their
        lifetime durations are given by the positions of the atoms in
        $\mathcal{M}_x$. 
\end{itemize}

The dual process with suppression $c \not\equiv 1$ can be coupled with the case
$c \equiv 1$. Given a realization of the process, if a branching occurs at time
$t$, the children are killed independently with probability $c(T-t)$.
(Note that as in the original CMJ process, suppression  translates into
trimming the dual tree.)

\begin{rem}
    We note that there are as many dual processes as there are point
    processes with intensity measure $\tau$. Here are a few natural choices:
    \begin{enumerate}
        \item Take $\mathcal{M}=\mathcal{P}$.
        \item Let $\mathcal{M}$ be a Poisson Point Process with intensity measure
            $\tau(\diff u)$. In this particular case, the dual process is a
            Bellman-Harris branching process (i.e., the offspring lifetime
            durations are independent conditional on offspring number). We note
            that $\tau(\diff u)$ appears naturally when considering the ancestral
            spine of a critical CMJ, see e.g.~\cite{schertzer2018height}. The
            measure $\tau$ can be obtained by size-biasing $\mathcal{P}$
            (i.e.\ biasing by the total mass of $\mathcal{P}$) and then
            recording the age of the individual at a uniformly chosen birth event.
    \end{enumerate}
\end{rem}
 
Let $(Y_t;\, t \le T)$ be the stochastic process valued in
$\cup_{n\in{\mathbb N}} {\mathbb R}_+^n$ recording the residual
life-times at time $t$ listed in increasing order, i.e.\ if $Y_t =
(Y_t^{(1)},\cdots, Y_t^{(n)})$ there are $n$ particles alive at time $t$
and $Y_t^{(k)}$ is the residual life-time of the $k$-th individual with
$Y_t^{(1)}<\cdots< Y_t^{(n)}$. (We assumed that $\tau$ has a density so
that the residual lifetimes are distinct a.s.) In particular, the
particle labelled $1$ at any given time $t$ will be the first to expire,
and at death time $t+Y_t^{(1)}$ a random number of children is produced.
We let $\dim(Y_t)$ denote the number of particules alive at time
$t$, i.e., the dimension of the vector $Y_t$. 
 
\begin{prop}[Backward Feynman-Kac formula]
    For any probability density $g$, we have
    \begin{equation}\label{eq:fk} 
        n(T,a) =
        \widehat{\mathbb E}_{a}\Bigg( \sum_{i\leq \dim(Y_T) }
            g(Y_T^{(i)}) \Bigg)
    \end{equation}
    where $(n(t,a);\, t, a \ge 0)$ is the unique solution to the
    McKendrick-von~Foerster equation started from $g$, and $\widehat
    {\mathbb E}_a$ is the distribution of the process $(Y_t;\, t \le T)$
    starting with an individual with residual lifetime $a$. 
\end{prop}

\begin{proof}
Let $t_1<\cdots<t_k <\cdots$ be the successive branching times of the
dual branching process. Since $\tau$ has a density, there is a single
branching particle at the successive branching times $t_1,\dots$ Define
the process
\[
    Z_s \coloneqq \sum_{i\leq \dim(Y_s)} n(T-s,Y_s^{(i)})
\]
See also Figure~\ref{F:dual} for a pictorial representation of
the process. It is plain from the definition that $n$ is preserved along
the characteristics of the PDE, i.e., that for every $x$ the function
$s\to n(T-s, x-s)$ remains constant on $[0,x)$. As a consequence,
$(Z_s;\, s\geq0)$ remains constant on every interval $[t_{n}, t_{n+1})$,
with the convention $t_0=0$. Define $z_n \coloneqq Z_{t_n}$ the value of
the process $(Z_t;\, t\geq0)$ at the $n$-th branching time. Let
$(\mathcal{F}_{n};\, n\in{\mathbb N})$ be the filtration induced by the
process $(z_n;\, n\in{\mathbb N})$. By definition of the dual
process, we have
\[
    Y_{t_n-}^{(i)} = Y_{t_{n-1}}^{(i)} - (t_n-t_{n-1})
    = Y_{t_{n-1}}^{(i)} - Y_{t_{n-1}}^{(1)}.
\]
For every $n>1$
\begin{align*}
    \widehat{\mathbb E}_a \left( z_n \ | \ \mathcal{F}_{n-1}\right) 
    &= \sum_{2 \leq i \leq \dim(z_{n-1}) } n(T-t_{n}, Y_{t_n-}^{(i)}) 
    +  c(T-t_n) \E\big[ \Angle{\mathcal{M}, n(T-t_n, \cdot)} \big] \\
    &= \sum_{2 \leq i \leq \dim(z_{n-1}) } n(T-t_{n}, Y_{t_n-}^{(i)}) 
    +  c(T-t_n) \int_0^{\infty} n(T-t_n, a) \tau(\diff a) \\
    &= \sum_{2 \leq i \leq \dim(z_{n-1}) } n(T-t_{n}, Y_{t_n-}^{(i)}) 
    +  n(T-t_n, 0) \\
    &= z_{n-1},
\end{align*}
where the third equality follows from the age boundary of the
McKendrick-von Foerster equation, and the last equality from the identity
$n(t-s, a-s) = n(t, a)$. As already mentioned, the process
$(Z_s;\, s\geq0)$ is constant between two branching times. As a
consequence, $(Z_s;\, s\geq0)$ is a martingale (w.r.t.\ its natural
fltration) so for every $s\geq0$,
\[ 
    n(T,a) = \widehat{\mathbb E}_{a}\bigg( \sum_{i\leq \dim(Y_s)} n(T-s, Y_s^{(i)}) \bigg). 
\]
Relation~\eqref{eq:fk} follows by taking $s=T$ in the latter expression.
\end{proof}

\begin{figure}[t]
    \centering
    \includegraphics[scale=0.6]{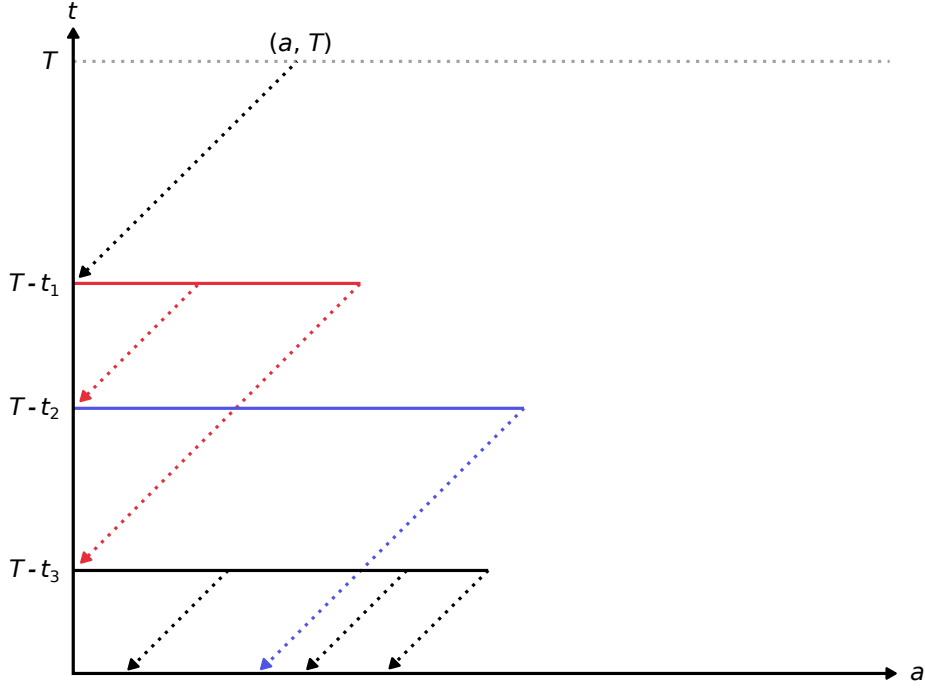}
    \caption{Graphical representation of the process $(Z_s;\, s \le T)$,
    where $s = T-t$. We start with a single individual with residual
    lifetime $a$. In this picture, time flows downwards for the branching
    process. The residual lifetime of the initial individual decreases
    linearly at speed one until reaching $0$ (this corresponds to time
    $T-t_1$ in our representation). At this time, the particle dies and
    produces 2 red particules. Residual lifetimes travel along the
    characteristics of the McKendrick-von Foerster PDE until reaching the
    spatial  boundary condition $\{a=0\}$ where a new branching occurs.}
    \label{F:dual}
\end{figure}

\section{Proofs of the main results}
\label{S:proofs}

In this section, we provide the proofs of the two laws of 
large numbers stated in Section~\ref{SS:mainResults}. The proof of
Theorem~\ref{thm:largePop} is a direct consequence of the results of the
previous section.

\begin{proof}[Proof of Theorem~\ref{thm:largePop}]
    Recall the definition of the empirical measure
    $\mu^N_t$. It can be written as
    \begin{equation}\label{eq:empirical-measure}
        \mu^N_t(\diff a \times \{i\}) = 
        \frac{1}{N} \sum_{k=1}^N \mu^{1,(k)}_t (\diff a \times \{i\})
    \end{equation}
    where $(\mu^{1,(k)}_t;\, k \ge 1)$ are independent copies of
    $\mu^{1}_t$. Let $f$ be an arbitrary continuous and bounded function
    on $\mathbb{R}_+\times \mathcal{S}$. The law of large numbers
    combined with Corollary~\ref{cor:forwardFK} implies that
    \[
        \Angle[\big]{\frac{1}{N} \mu^N_t, f} \longrightarrow  
        \Angle{\bar \mu_t, f} \quad \text{a.s.}
    \]
    which ends the proof.
\end{proof}

We now turn our attention to the proof of Theorem~\ref{thm:singlePop},
the law of large number started from a single individual. Let us briefly
recall the setting of this result.

We consider a sequence $( t_K;\, K \ge 0)$ of random times with $t_K \to
\infty$ a.s.\ on the non-extinction event.
We assume that the process starts from a single individual infected at
time $t=0$ and that the contact rate $c^K$ of the CMJ depends on $K$ in
the following way: $c^K(t) \coloneqq C(t-t_K)$, where $(C(t);\, t \in
\R)$ is a piecewise continuous function in $[0,1]$ such that $C(t)=1$ for
all $t < 0$.

\begin{proof}[Proof of Theorem~\ref{thm:singlePop}]
The result will follow by viewing the population at time $t_K+t$ as an
adequate random characteristic of the population at time $t_K$. Let us
recall some basic facts about random characteristics of CMJ processes
in our context. We refer to \cite{Jagers1984, Taib1992} for a more
thorough account on this notion.

We consider a plain CMJ with no contact rate or initial shifted law.
Every individual is characterized by an independent pair of random
variables $(\mathcal{P}_x, X_x)$. A random characteristic is a
real-valued stochastic process $(\chi(a);\, a \ge 0)$ that can be
constructed from the collection $(\mathcal{P}_x, X_x;\, x \in \mathscr{U})$.
(More formally it is a càdlàg process measurable w.r.t.\ the
$\sigma$-field induced by these variables.) By convention, it is extended
to a process defined on the whole real line by setting $\chi(a) \coloneqq
0$ for $a < 0$. 

For an individual $x \in \mathscr{U}$ let us write $\chi_x$ for the random
characteristic constructed from the collection $(\mathcal{P}_{xy},
X_{xy};\, y \in \mathscr{U})$. It is the characteristic constructed from the
tree rooted at $x$ of all descendants of $x$. The branching process
counted by the random characteristic $\chi$ is then defined as 
\[
    Z^\chi(t) = \sum_{x \in \mathscr{U}} \chi_x(t-\sigma_x).
\]
We now recall one of the main results of Jagers and Nerman
\cite{Jagers1984a}, namely Theorem~5.8 (see also Theorem~4, Appendix~A in
\cite{Taib1992}). Recall that $\alpha$ denotes the Malthusian parameter
defined in \eqref{eq:malthusian}. On top of all the assumptions above, we
make the two following extra assumptions 
\begin{enumerate}
    \item[(a)] The characteristic fulfills
        \[
            \sum_{n \ge 0} \sup_{n \le u \le n+1} e^{-\alpha u} 
            \E( \chi(u) ) < \infty.
        \]
    \item[(b)] The map $a \mapsto \mathbb{E}(\chi(a))$ is continuous a.e.\ with
        respect to the Lebesgue measure.
\end{enumerate}
Then there exists a positive r.v.\ $W_\infty$ (independent of the choice
of $\chi$) such that conditional on non-extinction 
\[
    Z^{\chi}(t)\exp(-\alpha t)  \to 
    W_\infty  \int_0^\infty \alpha e^{-\alpha a} \mathbb{E}(\chi(a)) \diff a 
    \quad \text{in probability as $t \to \infty$.}
\]

To illustrate the method, we recall that if we take 
$\chi(a)= \mathbbm{1}_{{\mathbb R}_+}(a)$ then $Z^{\chi}(t)$ coincides
with $Z(t)$, the total number of births before time $t$. For this
particular choice of (deterministic) characteristic, the two properties
above are immediately satisfied (recall that $\alpha > 0$), so that
conditional on non-extinction
\begin{equation*}
    \sum_{x} \mathbbm{1}_{\{\sigma_x<t\}} \exp(-\alpha t) \to W_\infty
    \quad  \text{in probability}.
\end{equation*}
This convergence ensures that the first item of our theorem is satisfied.

To prove the second item, let us set
\[
    \mathcal{P}_\emptyset = \sum_{i = 1}^{N_\emptyset} \delta_{A_i},
\]
for the atoms of the infection point process of the ancestor
$\mathcal{P}_\emptyset$. Further denote by 
$\mu^{(i)}_t$ the empirical measure of ages and compartments at time $t$
of the progeny of the $i$-th child of the ancestor, thinned by the
contact rate $(c(t);\, t \ge 0)$, i.e.,
\[
    \forall t \ge 0,\quad \mu^{(i)}_t 
    = \sum_{x \in \mathscr{U}} \mathbbm{1}_{\{ \widetilde{\sigma}_{ix} < t\}} 
        \delta_{(t - \widetilde{\sigma}_{ix}, X_{ix}(t-\widetilde{\sigma}_{ix}))},
\]
where $\widetilde{\sigma}_x$ refers to the infection time of $x$ once the
tree has been thinned. (That is, $\widetilde{\sigma}_x = \sigma_x$ if
$x$ remains infected after the thinning, or $\widetilde{\sigma} = \infty$
otherwise.) Our characteristic of interest can now be defined as
\[
    \chi^{(t,f)}(a) = f(a+t, X_\emptyset(a+t)) 
    + \sum_{i = 1}^{N_\emptyset} \mathbbm{1}_{\{A_i \in [a, a+t]\}} \Angle{\mu^{(i)}_{a+t}, f}.
\]
for a fixed time $t \ge 0$ and a fixed bounded continuous function $f$.
On the one hand, it should be clear that 
\begin{equation}
    Z^{\chi^{(t,f)}}(t_K) \overset{\mathrm{(d)}}{=} \Angle{\mu^K_{t_K+t}, f}.
\end{equation}
To see this, note that the process $Z^{\chi^{(t,f)}}$ is obtained from a
plain CMJ with no thinning, so that only infections after time $t_K$ are
removed due to the contact rate. 

On the other hand, $\chi^{(t,f)}(a)$ can be obtained by starting from an
initial individual with age $-a$, removing all atoms from its infection
point process before time $0$, and integrating the empirical measure of
the resulting CMJ at time $t$ against $f$. This is the description of the
CMJ with initial shifted law $\tilde{\mathscr{L}}$ conditional on $A =
a$, so that
\begin{equation} \label{eq:characShiftedLaw}
    \E\Big( \int_0^\infty g(a) \chi^{(t,f)}(a) \diff a \Big) =
    \Angle{\bar{\mu}_t, f},
\end{equation}
where in $\bar{\mu}_t$ the age of the initial ancestor is distributed
according to $g$.

Therefore, up to checking (a) and (b), Theorem~5.8 in
\cite{Jagers1984} shows that, as $K \to \infty$,
\[
    \Angle{\mu^K_{t_K+t}, f} \longrightarrow W_\infty \E\Big( \int_0^\infty \alpha
    e^{-\alpha a} \chi^{(t,f)}(a) \diff a\Big) 
\]
in probability, which in combination with \eqref{eq:characShiftedLaw}
proves the result.

All what remains to be shown is that (a) and (b) are fulfilled. That $a
\mapsto \E(\chi^{(t,f)}(a))$ is continuous a.e.\ follows directly from
the fact that $\tau$ has a density. Condition (b) is a consequence of the
following stochastic domination 
\[
    \chi^{(t,f)}(a) \le 
    \norm{f}_{\infty} \Big(1 + \sum_{i = 1}^{N_\emptyset} Z^{(i)}(t) \Big)
\]
where $(Z^{(i)}(t);\, t \ge 0)$ are i.i.d.\ copies of the CMJ without
thinning, independent of $\mathcal{P}_\emptyset$.
\end{proof}

\section{Inference procedure}  \label{S:inference}

In this section, we illustrate how to use our framework to make
inferences from macroscopic observables of the epidemic, e.g., incidence
of positively tested patients, hospital or ICU (intensive care unit)
admissions, deaths, \textit{etc}. We show how to use those observables to
extract the underlying age structure of the population, estimate model
parameters, and forecast the future of the epidemic.

We focused on a longitudinal case study in France. From March 18 2020,
the French government has provided daily reports of the numbers of ICU
and hospital admissions, of deaths, of discharged patients, and of
occupied ICU and hospital beds. Moreover, several theoretical studies
have already been conducted on the same dataset. This allowed us to fix
the values of some crucial biological parameters that had already been
estimated and to carry out a comparison with our method. We want to
emphasize that the aim of this section is to provide a mathematical
framework in which convergence results can be rigorously proved while
remaining flexible enough for other applications. Our goal is not to
provide new estimates of epidemiological parameters for France, as many
robust estimates are already available. For instance we do not provide
confidence intervals for our estimates, and neither do we conduct a
sensibility analysis.

The remainder of the section is laid out as follows. In
Section~\ref{SS:model} we identify the mathematical quantities that
impact the dynamics of the epidemic for large population sizes, and show
how to turn them into a likelihood. Section~\ref{SS:parametrization} then
presents the choice of distribution we made for these quantities and the
parameters that need to be estimated. Finally, estimation of these
parameters from the French incidence data is performed in
Section~\ref{SS:admission} and Section~\ref{SS:occupancy}. We start by
fitting a simple model in Section~\ref{SS:admission} and then show how
this model can be made more complex to account for more complex data in
Section~\ref{SS:occupancy}.

\subsection{Deriving the likelihood} \label{SS:model}

Under the assumptions of Theorem~\ref{thm:singlePop}, the number of
individuals in a given state $i$ at time $t$ converges to 
\begin{equation} \label{eq:prediction}
    n_i(t) = \int_0^\infty n(t,a) p(a,i) \diff a,
\end{equation}
where $(n(t,a);\, t,a \ge 0)$ is the solution to \eqref{eq:McKVF}. The
required assumptions are in essence that the epidemic has been ongoing
for a long enough time at the lockdown onset for the infected population
to be large, which we assume to hold true for France as the number of
cases on March 16 2020 was on the order of thousands of
individuals.

Therefore, we take \eqref{eq:prediction} as the predicted number of
individuals in state $i$ in our model. In order to turn
\eqref{eq:prediction} into a likelihood, we assume that the observed
number of individuals in state $i$ at time $t$ is distributed according
to some discrete law centered on the predicted value, which we take to be
a Poisson distribution. Then, the likelihood for the whole time period is
obtained by assuming that the observations are independent across states
and time. The explicit expression for the likelihood is provided in
Section~\ref{A:likExpression}. 

\begin{rem}
    The assumption that observations are independent is obviously not
    met. For instance, the number of occupied hospital beds is
    cumulative, so that any error is propagated from one day to the
    other. Moreover, there is a clear weekly effect in data that is not
    accounted for here. As deriving robust estimates is not the main
    purpose of this work, we prefer to keep this independence assumption
    that leads to simple expressions for the likelihood, while being
    aware of its limitation. This assumption could be relaxed by modeling
    explicitly the observation process and its potential errors.
\end{rem}

\begin{rem}
    We have decided to use an expression for the likelihood similar to
    that in \cite{salje_estimating_2020, sofonea_epidemiological_2020}.
    Such an expression only poorly accounts for the deviation of the
    stochastic model from its deterministic limit. Better accounting for
    this effect would require to use the likelihood of the stochastic
    model, or a Gaussian approximation of it obtained by deriving a
    functional central limit theorem. In our context the covariance
    structure of the population could be obtained by adapting the
    expressions in \cite[Section~3]{Jagers1984a} to incorporate the
    contact rate $(c(t);\, t \ge 0)$. However the resulting expressions
    would be quite cumbersome and computationally costly so that we
    prefer to use our simpler Poisson likelihood.
\end{rem}

Under our assumptions, the likelihood only depends on
\eqref{eq:prediction}, which is in turn determined by four quantities
that need to be parametrized:
\begin{enumerate}
    \item The intensity measure of the infection point process
        $(\tau(a);\, a \ge 0)$.
    \item The initial number of infected individuals and their age
        profile.
    \item The contact rate after lockdown $(c(t);\, t \ge 0)$.
    \item The one-dimensional marginals of the life-cycle process 
        $(p(a,i);\, a \ge 0)$ for $i \in \mathcal{S}$.
\end{enumerate}

\subsection{Parametrization of the model} \label{SS:parametrization}

\paragraph{Average infection measure.}
Recall the definition of $\tau$ and $R_0$ from
Section~\ref{SS:assumption}
and further define
\[
    \forall a \ge 0,\quad \hat{\tau}(a) = \frac{\tau(a)}{R_0}.
\]
The total mass of $\tau$, $R_0$, corresponds to the mean number of
secondary infections induced by a single infected individual if $c \equiv 1$. 
Thus $R_0$ is the basic reproduction number at the start of the epidemic,
when no control measure is enforced. In order to distinguish it from the
reproduction number during lockdown, it will be denoted by 
$R_{\mathrm{pre}}$. We leave it as a parameter to infer.

The function $(\hat{\tau}(a);\, a \ge 0)$ is the density of a probability
measure known as the generation time distribution \cite{wallinga07,
britton19}. This distribution has been estimated shortly after the
epidemic onset by several studies \cite{ferretti_quantifying_2020,
ganyani20, cereda2020early}. 
We use the estimation of \cite{ferretti_quantifying_2020}, and assume
that $\hat{\tau}$ is a Weibull distribution, that is 
\begin{equation} \label{eqn:tauhat}
    \forall a \ge 0,\quad \hat{\tau}(a) = 
    \frac{k}{\lambda} \Big(\frac{a}{\lambda}\Big)^{k-1} e^{-(a/\lambda)^k},
\end{equation}
where the values of the shape parameter $k$ and scale parameter $\lambda$
are recalled in Table~\ref{T:both}.

\paragraph{Initial condition.} According to Theorem~\ref{thm:singlePop},
the initial age structure of the population is 
\[
    \forall a \ge 0,\quad n(0, a) = W \alpha e^{-\alpha a},
\]
where $\alpha$ is the Malthusian parameter of the epidemic prior to
implementation of control measures, and $W$ is the number of infected
individuals at $t=0$, that is, at the lockdown onset. The parameter
$\alpha$ corresponds to the exponential growth rate of \emph{any}
observable of the epidemic during this period. We chose to estimate
it from the cumulative number of deaths, which appeared to be more
reliable than the number of detected cases as the number of tests
conducted in the early phase of the epidemic in France varied greatly.
It was estimated using a linear regression on the logarithm of the number
of deaths from March 7 to March 20 2020, and the corresponding basic
reproduction number before lockdown, $R_{\mathrm{pre}}$, was computed
using the Euler-Lotka equation \eqref{eq:malthusian} assuming that the
generation time distribution is given by \eqref{eqn:tauhat}. Both
estimates are shown in Table~\ref{T:both}.

{\renewcommand{\arraystretch}{1.5}
\begin{table}
    \centering
    \begin{tabular}{c p{8cm} c c}
        \toprule
        Notation &
        \multicolumn{1}{c}{Description} & Value & Source\\
        \midrule
        $\alpha$ & Pre-lockdown exponential growth rate & 0.315 & E\\
        $R_\mathrm{pre}$ & Basic reproduction number before lockdown & 3.25 & E\\
        $k$ & Shape parameter of the generation time & 2.83 
            & \cite{ferretti_quantifying_2020}\\
        $\lambda$ & Scale parameter of the generation time & 5.67
                  & \cite{ferretti_quantifying_2020}\\
        \bottomrule
    \end{tabular}
    \caption{Parameter values common to both models. In the ``Source''
        column, ``E'' indicates that the parameter has been estimated in the
        present work.}
    \label{T:both}
\end{table}
}

\paragraph{Contact rate.} The contact rate 
$(c(t);\, t \ge 0)$ accounts for the temporal variations in transmissions
after the lockdown onset. As we focus on the period from March to May
2020 where no additional control measure has been enforced in France, we
will assume that $(c(t);\, t \ge 0)$ is constant and denote by $c_0$ its
value, that is, $c \equiv c_0$. The reproduction number after the
lockdown is denoted by $R_\mathrm{post} \coloneqq c_0 R_\mathrm{pre}$.

\paragraph{Life-cycle.} The last quantities that need to be defined are
the one-dimensional marginals of the life-cycle process $(X(a);\, a \ge
0)$. These could be directly estimated from hospital patient pathways as
in \cite{Linton2020, verity_estimates_2020, lefranck21}. However, when
such data is not available they need to be estimated from individual
counts in each compartment. In this case, we propose the following
parametrization of the process $(X(a);\, a \ge 0)$ based on
Gamma-distributed sojourn times. 

Let us denote by $(X_n;\, n \ge 0)$ the sequence of states visited by
$(X(a);\, a \ge 0)$. We assume that $(X_n;\, n \ge 0)$ is a Markov chain
on $\mathcal{S}$, and that it ends either in a ``dead'' or ``recovered''
state, that are assumed to be absorbing. 

For $i \in \mathcal{S}$, the sojourn time in $i$ is supposed to be
Gamma-distributed with mean $m_i$ and variance $m_i / \gamma$, for some
global dispersion parameter $\gamma$ shared across all states. More
precisely, let $(D_n;\, n \ge 0)$ denote the sequence of sojourn times of
$(X(a);\, a \ge 0)$, that is, $D_n$ is the sojourn time in state $X_n$.
We assume that conditional on $(X_n;\, n \ge 0)$, the variables $(D_n;\,
n \ge 0)$ are independent. Moreover, if $X_n = i_n$, then $D_n$ follows a
Gamma distribution with mean $m_{i_n}$ and variance $m_{i_n}/\gamma$,
that is,
\[
    D_n \sim \frac{\gamma^{\gamma m_{i_n}}}{\Gamma(\gamma m_{i_n})} 
    u^{\gamma m_{i_n} - 1} e^{-\gamma u} \diff u.
\]

Thus, the one-dimensional marginals are parametrized by the transitions
of a Markov chain $(X_n;\, n \ge 0)$ on $\mathcal{S}$, as well as by one
parameter $m_i$ for each $i \in \mathcal{S}$, and a global parameter
$\gamma$. Under this parametrization the one-dimensional marginals can be
efficiently computed, while only requiring one parameter for the sojourn
time in each state. Two concrete examples of Markov chains $(X_n;\, n \ge
0)$ are discussed in the next sections.

\subsection{Inference with the admission model} \label{SS:admission}

The first model that we consider, the admission model, is a parsimonious 
model designed to obtain estimates of the reproduction number during
lockdown, and of the number of infections in France in early March. 
It is illustrated in Figure~\ref{F:admissionModel}. We fit it to the three
``incidence'' time series: the daily number of admissions in hospital and
ICU, and the daily number of deaths.

\paragraph{Description of the model.}
Upon infection, with probability $1-p_\mathrm{hosp}$, an individual
develops a mild form of  COVID-19 and is placed in state $I$, which
encompasses all cases that do not require a hospitalization. With
probability $p_{\mathrm{hosp}}$ the individual has a severe infection and
is placed in state $C$. Individuals in state $C$ are eventually
hospitalized and moved to state $H$. Then, with probability $p_{\icu}$
individuals in state $H$ are admitted in ICU and moved to state $U$.
Otherwise they eventually recover and are discharged. Finally,
individuals in state $U$ die with probability $p_{\mathrm{death}}$, or
recover with probability $1-p_{\mathrm{death}}$. In this model, only
individuals in ICU may die.

\begin{figure}
    \begin{center}
        \includegraphics[width=.95\textwidth]{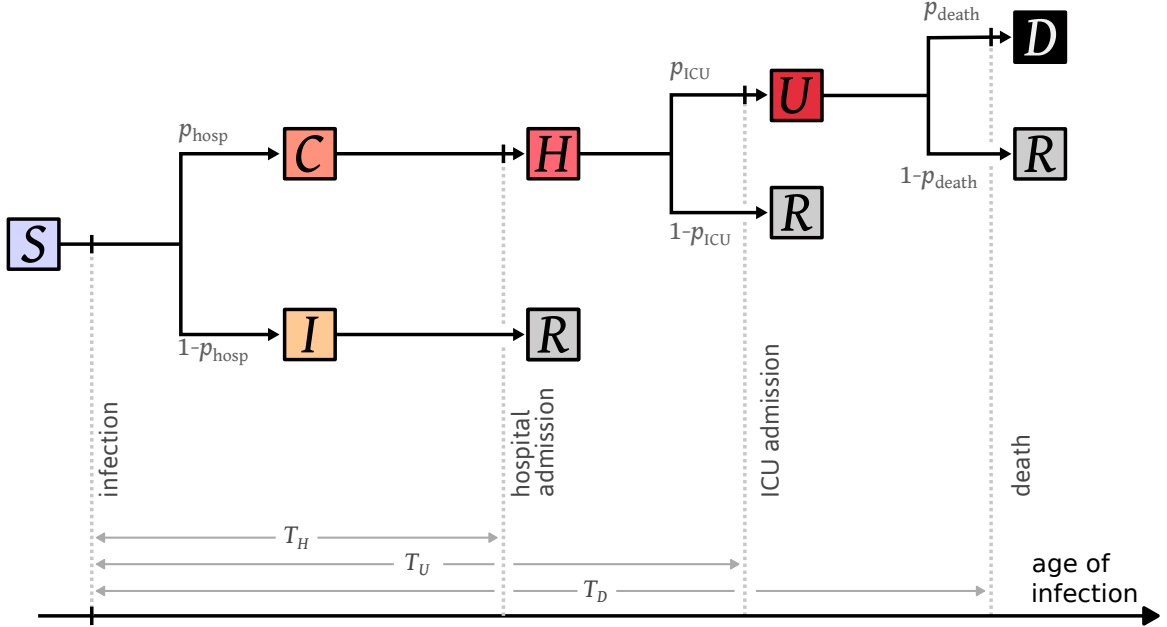}
    \end{center}
    \caption{Illustration of the admission model.}
    \label{F:admissionModel}
\end{figure}

As we are fitting the number of individuals that enter a state, and not
the number of individuals that are currently in that state, we only need
to track the times $T_H$, $T_U$, and $T_D$ elapsed between infection and
hospital admission, ICU admission, and death, respectively.

\paragraph{Inference results.}
Estimations of $p_\hosp$, $p_\icu$ and of the death probability
conditional on hospitalization (equal in our setting to $p_\icu \times
p_{\mathrm{death}}$)  in France have already been conducted
in~\cite{salje_estimating_2020}. We used these estimates and considered
the values of $p_\hosp$, $p_\icu$, and $p_\mathrm{death}$ as fixed. All
other parameters were estimated using a maximum likelihood procedure
described in Section~\ref{A:likExpression}. The parameter estimations are
provided in Table~\ref{T:admission}, and the corresponding predicted
values for the time series under consideration are displayed in
Figure~\ref{F:admissionFit}. Overall, our simple model seems to match the
observed data. Note however that the model overestimates the number of
ICU admissions in the second part of the lockdown. This is likely due to
a temporal reduction in the ICU admission probability which has been
reported in~\cite{salje_estimating_2020}.

Our estimation of the basic reproduction number during the lockdown
period is $R_{\mathrm{post}} = 0.745$. This suggests that lockdown has
reduced the basic reproduction number by a factor $c_0 = 0.23$ compared to the
beginning of the epidemic. Moreover, we estimated that $W = 9.85 \times
10^5$ infections have occurred in France before March 17. Both these
values are in line with previous estimates for
France~\cite{sofonea_epidemiological_2020, salje_estimating_2020}. 

We did not impose that $T_H < T_U$ in the inference procedure.
Interestingly we found that the data are best explained by assuming that 
the mean of $T_H$ is $14.4$ days, whereas the mean of $T_U$ is $11.4$
days. This indicates that the delay between infection and hospital
admission is shorter for individuals that end up in ICU, compared to the
average time between infection and hospitalization. Therefore it would be
more appropriate to allow individuals to have an admission to hospital
delay that is different depending on whether they will end up in ICU or
not, modeling the fact that they have a more severe form of the disease.
We estimated the mean of $T_D$, the time between infection and death, to
be $18.6$ days. This estimate is lower than but consistent with
previous estimates based on the study of individual-case
data~\cite{Wu2020, Linton2020, verity_estimates_2020}.

{
\renewcommand{\arraystretch}{1.5}
\begin{table}[t]
    \centering
    \begin{tabular}{c p{8cm} c c}
        \toprule
        Notation & \multicolumn{1}{c}{Description} & Value & Source\\
        \midrule
        $R_{\mathrm{post}}$ & Reproduction number during lockdown &
        $0.745$ & E\\
        $W$ & Total number of infections before March 17 2020 & $9.85 \times
        10^5$ & E\\
        $p_{\mathrm{hosp}}$ & Probability of being hospitalized & $0.036$
                            & \cite{salje_estimating_2020}\\
        $p_{\mathrm{ICU}}$ & Probability of entering ICU conditional on
        being at the hospital & $0.19$ & \cite{salje_estimating_2020}\\
        $p_\mathrm{ICU}\cdot p_{\mathrm{death}}$ & Death probability
        conditional on being hospitalized& $0.181$ &
        \cite{salje_estimating_2020}\\
        $T_H$ & Delay between infection and hospital admission& $14.4$
        days & E\\
        $T_U$ & Delay between infection and ICU admission& $11.4$ days &
        E\\
        $T_D$ & Delay between infection and death & $18.6$ days & E\\
        $\gamma$ & Scale parameter common to all Gamma distributions &
        0.463 & E\\
        \bottomrule
    \end{tabular}
    \caption{Inferred parameter set for the admission model.
        The values indicated for the durations correspond to the means of
        the Gamma distributions. In the ``Source'' column, ``E''
        indicates that the parameter has been estimated in the current work.}
    \label{T:admission}
\end{table}
}

\subsection{Inference with the occupancy model} \label{SS:occupancy}

We now consider a model aimed at providing predictions for the number of
hospitalized individuals and ICU patients. The model is fitted to the
three ``incidence'' time-series, and to three additional ``prevalence'' 
time-series: the number of occupied hospital and ICU beds, and the number
of discharged hospital patients.

A first attempt to fit the prevalence curves could be to keep the
admission model of Figure~\ref{F:admissionModel} and to estimate the time
between hospital admission and discharge using the observed number of
occupied ICU, hospital beds, and discharged patients. However this only
yields a poor fit of the data (see Section~\ref{A:poorFit}). We
identified two main reasons for this discrepancy. First, we assumed that
all individuals are admitted to ICU prior to death. Using the probability
estimated in~\cite{salje_estimating_2020} then yields that the
probability of dying conditional on being in ICU is $0.953$. This value
is unrealistically high, and we need to assume that a fraction of
hospital deaths occur without going through the ICU. Second, under the
admission model, the delay between hospital admission and discharge is
almost unimodal. However, the observed number of occupied hospital beds
rises fast but falls slowly. Such a shape cannot be easily accounted for
by a unimodal distribution for the time spent in hospital. 

\paragraph{Description of the model.}
Taking into account the previous two points required us to make the model
more complex. The resulting model, referred to as the \emph{occupancy model}, is
illustrated in Figure~\ref{F:occupancyModel}. We now consider that upon
infection, individuals go to one of three states depending on the
severity of their infection:
\begin{itemize}
    \item The state $C_u$ which gathers critical infections that lead to death
        or ICU admission. The probability of having a critical 
        infection is denoted by $p_\mathrm{crit}$.
    \item The state $C_h$ which corresponds to severe infections that require a
        hospitalization but are not critical. Such infections occur with 
        probability $p_\mathrm{sev}$.
    \item The $I$ state which consists of all mild infections that do not
        lead to a hospital admission, and occur with probability
        $1-p_\mathrm{crit}-p_\mathrm{sev}$.
\end{itemize}

Individuals in state $C_h$ are admitted to hospital after a duration
$D_{C_h}$. Then, with probability $p_{\mathrm{short}}$ they are discharged after a duration $D_\mathrm{short}$, while with probability
$1-p_\mathrm{short}$ they are discharged after a duration $D_\mathrm{long}$.

Critically infected individuals are admitted to hospital after a duration
$D_{C_u}$. Upon arrival at hospital, they die immediately with probability
$d_\mathrm{hosp}$, or go to ICU after a duration $D_{H_u}$. Individuals in ICU
die with probability $d_\mathrm{ICU}$ after a delay $D_D$. Otherwise they
are discharged after a stay of length $D_U$.

\begin{figure}
    \begin{center}
        \includegraphics[width=.95\textwidth]{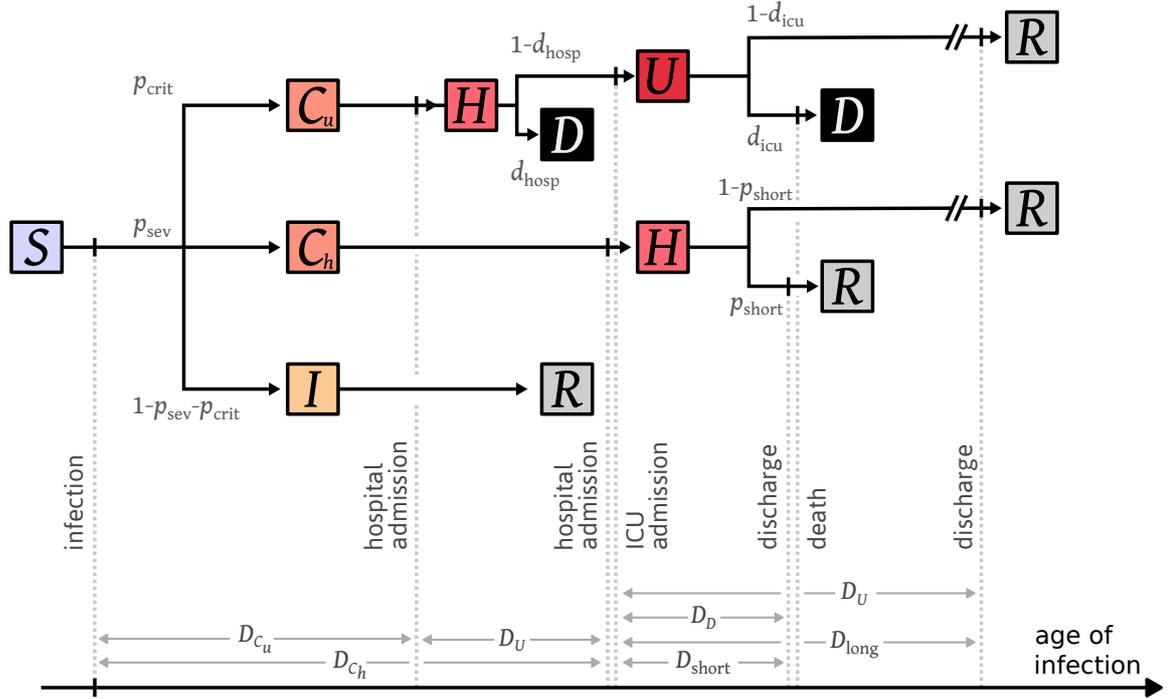}
    \end{center}
    \caption{Illustration of the occupancy model}
    \label{F:occupancyModel}
\end{figure}

\paragraph{Inference results.}
In our model, the probability of hospital admission is
$p_\mathrm{crit}+p_\mathrm{sev}$, the probability of ICU admission is
$p_\mathrm{crit}(1-d_\mathrm{hosp})$ and that of death is
$p_\mathrm{crit}(d_\mathrm{hosp} + (1-d_\mathrm{hosp})d_\mathrm{ICU})$.
We have fixed these three values to those estimated
in~\cite{salje_estimating_2020}, and we only had one remaining
parameter out of 4 ($p_\mathrm{crit}$, $p_\mathrm{sev}$,  $d_\mathrm{short}$, $d_\mathrm{ICU}$) to estimate from the data. We have fixed the time $D_U$ to
$1.5$ days as estimated in~\cite{salje_estimating_2020}. All other
parameters were estimated using a maximum likelihood method described in
Section~\ref{A:likExpression}. The estimated parameter set is shown in
Table~\ref{T:occupancy}, while Figure~\ref{F:occupancyFit} shows the
best-fitting model. 

The estimated parameters provide a good fit of the six observed
time series. Again, the model has a tendency to overestimate the ICU
admissions in the second part of the lockdown, which has the same
interpretation as for the admission model.

Under the occupancy model, we estimated that $R_{\mathrm{post}} = 0.734$,
and $W = 9.52\times 10^5$. These estimates are extremely close to those
made with the admission model. The estimated mean time between infection
and death or hospital, ICU admission are respectively $19.5$ days, $13.7$
days and $12.5$ days. Again we see that these estimates in the more
complex model are consistent with those of the simple model. The mean
recovery time from hospital is $19.4$ days for severe infections, and
$28.2$ days for critical infections. This yields an overall mean recovery
time of $20.0$ days. Finally, we estimated that the death probability
conditional on being in ICU is $0.709$. This yields that in our model a
fraction $0.256$ of all deaths occur shortly after hospital admission.
This result is consistent with~\cite{salje_estimating_2020} that
estimated that a fraction $0.15$ of all deaths occurred within the first
day after hospital admission. However, it has been reported
in~\cite{rapport_hebdo} that the death probability of ICU patients is
$0.23$. Our estimated value is thus unrealistically high. This indicates
that there is a fraction of hospital deaths that occur without any ICU
admission, and not quickly after hospital admission, that our model is
not accounting for.

{
\renewcommand{\arraystretch}{1.5}
\begin{table}
    \centering
    \begin{tabular}{c p{8cm} c c}
        \toprule
        Notation & \multicolumn{1}{c}{Description} & Value & Source\\
        \midrule
        $R_{\mathrm{post}}$ & Reproduction number during lockdown & $0.734$ & E\\
        $W$ & Total number of infections before March 17 2020 & $9.52 \times
        10^5$ & E\\
        $p_\mathrm{crit}+p_\mathrm{sev}$ & Probability of being
        hospitalized & $0.036$ & \cite{salje_estimating_2020}\\
        $\frac{p_\mathrm{crit} (1-d_\mathrm{hosp})}{p_\mathrm{crit}+p_\mathrm{sev}}$ 
        & Probability of entering ICU conditional on being at the
        hospital & $0.19$ & \cite{salje_estimating_2020}\\
        $\frac{d_\mathrm{hosp}+(1-d_\mathrm{hosp})d_\mathrm{ICU}}%
        {1+p_\mathrm{sev}/p_\mathrm{crit}}$ 
        & Death probability conditional on being hospitalized& $0.181$ &
        \cite{salje_estimating_2020}\\
        $d_\mathrm{ICU}$ & Probability of death conditional on being in
        ICU& $0.709$ & E\\
        $p_\mathrm{short}$ & Probability of a short stay at hospital & 
        $0.701$ & E\\
        $D_{C_h}$ & Delay between severe infection and hospital
        admission& $14.5$ days & E\\
        $D_\mathrm{short}$ & Delay between hospital admission and quick
        discharge& $7.36$ days & E\\
        $D_\mathrm{long}$ & Delay between hospital admission and slow
        discharge& $47.5$ days & E\\
        $D_{C_u}$ & Delay between critical infection and hospital
        admission& $11.0$ days & E\\
        $D_H$ & Delay between hospital admission and ICU admission &
        $1.5$ days & \cite{salje_estimating_2020}\\
        $D_U$ & Delay between ICU admission and discharge& $28.2$ days
              & E\\
        $D_D$ & Delay between ICU admission and death & $9.90$ days & E\\
        $\gamma$ & Scale parameter common to all Gamma distributions &
        $0.316$ & E\\
        \bottomrule
    \end{tabular}
    \caption{Inferred parameter set for the occupancy model.
        The values indicated for the durations correspond to the means of
        the Gamma distributions. In the ``Source'' column, ``E''
        indicates that the parameter has been estimated in the current work. }
        \label{T:occupancy}
\end{table}
}

Our estimates, though they are not the key message of the present paper,
can nevertheless draw attention to potential heterogeneities in the
infected population. We estimated that the mean time between infection
and ICU admission is shorter than that between infection and hospital
admission. This suggests that the time between infection and severe
symptom onset is shorter for critical infection, that lead to ICU
admission, than for milder ones. Moreover, fitting the prevalence time
series required to divide the hospital and death compartments in two
subcompartments, indicating that the data are not well explained by a
simple homogeneous model, as seen in Figure~\ref{F:poorFit}. Such heterogeneity
could originate from underlying structuring variables, such as
comorbidity or (actual) age, that we are not accounting for. Many estimates
of clinical features, such as the incubation period, are obtained from a
pooled dataset that does not take heterogeneity in the population into
account~\cite{Backer2020, Linton2020, Lauer2020, Tindale2020, Bi2020,
massonnaud_covid_2020, djidjou_optimal_2020}. When estimating the total
number of infected individuals using only a fraction of the detected
cases, e.g., using the hospital admissions or deaths, it is interesting
to keep in mind that the time periods estimated from pooled studies
could be inaccurate for the fraction of infected individuals under
consideration.

\subsection*{Acknowledgements}

The authors thank the \emph{Center for Interdisciplinary Research in
Biology} (CIRB) for funding and the taskforce MODCOV19 of the INSMI
(CNRS) for their technical/scientific support during the epidemic. P.C.
has received funding from the European Union's Horizon 2020 research and
innovation program under the Marie Sk{\l o}dowska-Curie grant agreement
PolyPath 844369.


\bibliographystyle{plain}
\bibliography{ms_jomb_covid.bib}

\begin{thebibliography}{10}

\bibitem{Backer2020}
Jantien~A. Backer, Don Klinkenberg, and Jacco Wallinga.
\newblock Incubation period of 2019 novel coronavirus {(2019-nCoV)} infections
  among travellers from {Wuhan}, {China}, {20--28} {January} 2020.
\newblock {\em Eurosurveillance}, 25(5), 2020.

\bibitem{baker18}
J.~Baker, P.~Chigansky, K.~Hamza, and F.~C. Klebaner.
\newblock Persistence of small noise and random initial conditions.
\newblock {\em Advances in Applied Probability}, 50(A):67--81, 2018.

\bibitem{ball1986unified}
Frank Ball.
\newblock A unified approach to the distribution of total size and total area
  under the trajectory of infectives in epidemic models.
\newblock {\em Advances in Applied Probability}, 18:289--310, 1986.

\bibitem{barbour16}
A.~D. Barbour, P.~Chigansky, and F.~C. Klebaner.
\newblock On the emergence of random initial conditions in fluid limits.
\newblock {\em Journal of Applied Probability}, 53:1193--1205, 2016.

\bibitem{barbour1975duration}
Andrew~D Barbour.
\newblock The duration of the closed stochastic epidemic.
\newblock {\em Biometrika}, 62:477--482, 1975.

\bibitem{Bi2020}
Qifang Bi, Yongsheng Wu, Shujiang Mei, Chenfei Ye, Xuan Zou, Zhen Zhang,
  Xiaojian Liu, Lan Wei, Shaun~A. Truelove, Tong Zhang, Wei Gao, Cong Cheng,
  Xiujuan Tang, Xiaoliang Wu, Yu~Wu, Binbin Sun, Suli Huang, Yu~Sun, Juncen
  Zhang, Ting Ma, Justin Lessler, and Tiejian Feng.
\newblock Epidemiology and transmission of covid-19 in 391 cases and 1286 of
  their close contacts in shenzhen, china: a retrospective cohort study.
\newblock {\em The Lancet Infectious Diseases}, 20:911--919, 2020.

\bibitem{brauer05}
Fred Brauer.
\newblock The {Kermack--McKendrick} epidemic model revisited.
\newblock {\em Mathematical Biosciences}, 198:119--131, 2005.

\bibitem{brauer19}
Fred Brauer, Carlos Castillo-Chavez, and Zhilan Feng.
\newblock {\em Mathematical Models in Epidemiology}.
\newblock Texts in Applied Mathematics. Springer-Verlag New York, 2019.

\bibitem{britton19}
Tom Britton and Gianpaolo Scalia~Tomba.
\newblock Estimation in emerging epidemics: biases and remedies.
\newblock {\em Journal of The Royal Society Interface}, 16:20180670, 2019.

\bibitem{cereda2020early}
Danilo Cereda, Mattia Manica, Marcello Tirani, Francesca Rovida, Vittorio
  Demicheli, Marco Ajelli, Piero Poletti, Filippo Trentini, Giorgio Guzzetta,
  Valentina Marziano, Raffaella Piccarreta, Antonio Barone, Michele Magoni,
  Silvia Deandrea, Giulio Diurno, Massimo Lombardo, Marino Faccini, Angelo Pan,
  Raffaele Bruno, Elena Pariani, Giacomo Grasselli, Alessandra Piatti, Maria
  Gramegna, Fausto Baldanti, Alessia Melegaro, and Stefano Merler.
\newblock The early phase of the {COVID-19} epidemic in {Lombardy}, {Italy}.
\newblock {\em Epidemics}, 37:100528, 2021.

\bibitem{Crawford}
Dorothy~H. Crawford.
\newblock {\em Deadly companions: how microbes shaped our history}.
\newblock Oxford University Press, second edition, 2018.

\bibitem{DiDomenico2020}
Laura Di~Domenico, Giulia Pullano, Chiara~E. Sabbatini, Pierre-Yves Bo{\"e}lle,
  and Vittoria Colizza.
\newblock Impact of lockdown on {COVID-19} epidemic in {\^i}le-de-france and
  possible exit strategies.
\newblock {\em BMC Medicine}, 18:240, 2020.

\bibitem{diekmann1977}
Odo Diekmann.
\newblock Limiting behaviour in an epidemic model.
\newblock {\em Nonlinear Analysis: Theory, Methods \& Applications},
  1:459--470, 1977.

\bibitem{Diekmann2013}
Odo Diekmann, Hans Heesterbeek, and Tom Britton.
\newblock {\em Mathematical Tools for Understanding Infectious Disease
  Dynamics}.
\newblock Princeton Series in Theoretical and Computational Biology. Princeton
  University Press, 2013.

\bibitem{Diekmann1990}
Odo Diekmann, J.~A.~P. Heesterbeek, and J.~A.~J. Metz.
\newblock On the definition and the computation of the basic reproduction ratio
  {$R_0$} in models for infectious diseases in heterogeneous populations.
\newblock {\em Journal of Mathematical Biology}, 28:365--382, 1990.

\bibitem{diekmann1995}
Odo Diekmann, J.A.J. Metz, and J.A.P. Heesterbeek.
\newblock The legacy of {Kermack} and {McKendrick}.
\newblock In D.~Mollison, editor, {\em Epidemic Models: Their Structure and
  Relation to Data}. Cambridge University Press, Cambridge, 1995.

\bibitem{djidjou_optimal_2020}
R.~Djidjou-Demasse, Y.~Michalakis, M.~Choisy, M.~T. Sofonea, and S.~Alizon.
\newblock Optimal covid-19 epidemic control until vaccine deployment.
\newblock {\em medRxiv}, 2020.

\bibitem{duchamps2021general}
Jean-Jil Duchamps, F{\'e}lix Foutel-Rodier, and Emmanuel Schertzer.
\newblock General epidemiological models: Law of large numbers and contact
  tracing.
\newblock 2021.

\bibitem{evgeniou_epidemic_2020}
Theodoros Evgeniou, Mathilde Fekom, Anton Ovchinnikov, Raphael Porcher, Camille
  Pouchol, and Nicolas Vayatis.
\newblock Epidemic models for personalised {COVID-19} isolation and exit
  policies using clinical risk predictions.
\newblock {\em Production and Operations Management, Special Issue on Managing
  Pandemics: A POM Perspective}, 2021.

\bibitem{fan2020convergence}
Jie~Yen Fan, Kais Hamza, Peter Jagers, and Fima Klebaner.
\newblock Convergence of the age structure of general schemes of population
  processes.
\newblock {\em Bernoulli}, 26:893--926, 2020.

\bibitem{ferretti_quantifying_2020}
Luca Ferretti, Chris Wymant, Michelle Kendall, Lele Zhao, Anel Nurtay, Lucie
  Abeler-D{\"o}rner, Michael Parker, David Bonsall, and Christophe Fraser.
\newblock Quantifying {SARS-CoV-2} transmission suggests epidemic control with
  digital contact tracing.
\newblock {\em Science}, 368(6491), 2020.

\bibitem{ferriere2009stochastic}
R\'egis Ferri\`ere and Viet~Chi Tran.
\newblock Stochastic and deterministic models for age-structured populations
  with genetically variable traits.
\newblock {\em ESAIM: Proc.}, 27:289--310, 2009.

\bibitem{forien2021epidemic}
Rapha{\"e}l Forien, Guodong Pang, and {\'E}tienne Pardoux.
\newblock Epidemic models with varying infectivity.
\newblock {\em SIAM Journal on Applied Mathematics}, 81:1893--1930, 2021.

\bibitem{forien2021}
Rapha{\"e}l Forien, Guodong Pang, and {\'E}tienne Pardoux.
\newblock Estimating the state of the {COVID-19} epidemic in {France} using a
  model with memory.
\newblock {\em Royal Society Open Science}, 8:202327, 2021.

\bibitem{rapport_hebdo}
Sant{\'e}~Publique France.
\newblock {COVID-19} : point {\'e}pid{\'e}miologique du 4 juin 2020, 2020.

\bibitem{donnee_france}
Sant{\'e}~Publique France.
\newblock Donn{\'e}es hospitali{\`e}res relatives {\`a} l'{\'e}pid{\'e}mie de
  {COVID-19}, 2020.

\bibitem{ganyani20}
Tapiwa Ganyani, C{\'e}cile Kremer, Dongxuan Chen, Andrea Torneri, Christel
  Faes, Jacco Wallinga, and Niel Hens.
\newblock Estimating the generation interval for coronavirus disease
  ({COVID-19}) based on symptom onset data, {March 2020}.
\newblock {\em Eurosurveillance}, 25, 2020.

\bibitem{gaubert20}
St\'ephane Gaubert, Marianne Akian, Xavier Allamigeon, Marin Boyet, Baptiste
  Colin, Th\'eotime Grohens, Laurent Massouli\'e, David~P. Parsons,
  Fr\'ed\'eric Adnet, \'Erick Chanzy, Laurent Goix, Fr\'ed\'eric Lapostolle,
  \'Eric Lecarpentier, Christophe Leroy, Thomas Loeb, Jean-S\'ebastien Marx,
  Caroline T\'elion, Laurent Tr\'eluyer, and Pierre Carli.
\newblock Understanding and monitoring the evolution of the {Covid-19} epidemic
  from medical emergency calls: the example of the {Paris} area.
\newblock {\em Comptes Rendus. Math\'ematique}, 358:843--875, 2020.

\bibitem{hamza2013age}
Kais Hamza, Peter Jagers, and Fima~C Klebaner.
\newblock The age structure of population-dependent general branching processes
  in environments with a high carrying capacity.
\newblock {\em Proceedings of the Steklov Institute of Mathematics},
  282:90--105, 2013.

\bibitem{inaba17}
Hisashi Inaba.
\newblock {\em Age-Structured Population Dynamics in Demography and
  Epidemiology}.
\newblock Springer Singapore, Singapore, 2017.

\bibitem{jagers1975branching}
Peter Jagers.
\newblock {\em Branching processes with biological applications}.
\newblock Wiley, 1975.

\bibitem{jagers2000population}
Peter Jagers and Fima~C Klebaner.
\newblock Population-size-dependent and age-dependent branching processes.
\newblock {\em Stochastic Processes and their Applications}, 87(2):235--254,
  2000.

\bibitem{jagers2011population}
Peter Jagers and Fima~C Klebaner.
\newblock Population-size-dependent, age-structured branching processes linger
  around their carrying capacity.
\newblock {\em Journal of Applied Probability}, 48:249--260, 2011.

\bibitem{Jagers1984a}
Peter Jagers and Olle Nerman.
\newblock The growth and composition of branching populations.
\newblock {\em Advances in Applied Probability}, 16:221--259, 1984.

\bibitem{Jagers1984}
Peter Jagers and Olle Nerman.
\newblock Limit theorems for sums determined by branching and other
  exponentially growing processes.
\newblock {\em Stochastic Processes and their Applications}, 17:47--71, 1984.

\bibitem{kermack1927}
William~O. Kermack and Anderson~G. McKendrick.
\newblock A contribution to the mathematical theory of epidemics.
\newblock {\em Proceedings of the Royal Society of London. Series A, Containing
  Papers of a Mathematical and Physical Character}, 115(772):700--721, 1927.

\bibitem{Kim2006}
Mi-Young Kim.
\newblock Long-time stability of numerical solutions to {Gurtin--MacCamy}
  equations by method of characteristics.
\newblock {\em Applied Mathematics and Computation}, 176:552--562, 2006.

\bibitem{Krauss2003}
H.~Krauss.
\newblock {\em Zoonoses : infectious diseases transmissible from animals to
  humans}.
\newblock ASM Press, 2003.

\bibitem{Lauer2020}
Stephen~A Lauer, Kyra~H Grantz, Qifang Bi, Forrest~K Jones, Qulu Zheng,
  Hannah~R Meredith, Andrew~S Azman, Nicholas~G Reich, and Justin Lessler.
\newblock The incubation period of coronavirus disease 2019 {(COVID-19)} from
  publicly reported confirmed cases: estimation and application.
\newblock {\em Annals of internal medicine}, 172:577--582, 2020.

\bibitem{lefranck21}
No{\'e}mie Lefrancq, Juliette Paireau, Nathana{\"e}l Hoz{\'e}, No{\'e}mie
  Courtejoie, Yazdan Yazdanpanah, Lila Bouadma, Pierre-Yves Bo{\"e}lle, Fanny
  Chereau, Henrik Salje, and Simon Cauchemez.
\newblock Evolution of outcomes for patients hospitalised during the first 9
  months of the {SARS-CoV-2} pandemic in {France}: A retrospective national
  surveillance data analysis.
\newblock {\em The Lancet Regional Health - Europe}, 5:100087, 2021.

\bibitem{Linton2020}
Natalie~M. Linton, Tetsuro Kobayashi, Yichi Yang, Katsuma Hayashi, Andrei~R.
  Akhmetzhanov, Sung-mok Jung, Baoyin Yuan, Ryo Kinoshita, and Hiroshi
  Nishiura.
\newblock Incubation period and other epidemiological characteristics of 2019
  novel coronavirus infections with right truncation: A statistical analysis of
  publicly available case data.
\newblock {\em Journal of Clinical Medicine}, page 538, 2020.

\bibitem{massonnaud_covid_2020}
Cl{\'e}ment Massonnaud, Jonathan Roux, and Pascal Cr{\'e}pey.
\newblock {COVID-19}: Forecasting short term hospital needs in {France}.
\newblock {\em medRxiv}, 2020.

\bibitem{metz1978}
J.~A.~J. Metz.
\newblock The epidemic in a closed population with all susceptibles equally
  vulnerable; some results for large susceptible populations and small initial
  infections.
\newblock {\em Acta Biotheoretica}, 27:75--123, 1978.

\bibitem{Metz1986}
J.~A.~J. Metz and Odo Diekmann.
\newblock {\em The Dynamics of Physiologically Structured Populations}.
\newblock Lecture Notes in Biomathematics. Springer, Berlin, Heidelberg, 1986.

\bibitem{Nerman1981}
Olle Nerman.
\newblock On the convergence of supercritical general {(C-M-J)} branching
  processes.
\newblock {\em Zeitschrift f{\"{u}}r Wahrscheinlichkeitstheorie und Verwandte
  Gebiete}, 57:365--395, 1981.

\bibitem{pang2020functional}
Guodong Pang and Etienne Pardoux.
\newblock Functional limit theorems for non-{Markovian} epidemic models.
\newblock 2020.

\bibitem{pang2021functional}
Guodong Pang and Etienne Pardoux.
\newblock Functional law of large numbers and pdes for epidemic models with
  infection-age dependent infectivity.
\newblock 2021.

\bibitem{reddingius1971}
Joannes Reddingius.
\newblock Notes on the mathematical theory of epidemics.
\newblock {\em Acta Biotheoretica}, 20:125--157, 1971.

\bibitem{roques_using_2020}
Lionel Roques, Etienne~K Klein, Julien Papaix, Antoine Sar, and Samuel
  Soubeyrand.
\newblock Using early data to estimate the actual infection fatality ratio from
  {COVID-19} in {France}.
\newblock {\em Biology}, 9:97, 2020.

\bibitem{salje_estimating_2020}
Henrik Salje, C{\'e}cile Tran~Kiem, No{\'e}mie Lefrancq, No{\'e}mie Courtejoie,
  Paolo Bosetti, Juliette Paireau, Alessio Andronico, Nathana{\"e}l Hoz{\'e},
  Jehanne Richet, Claire-Lise Dubost, Yann Le~Strat, Justin Lessler, Daniel
  Levy-Bruhl, Arnaud Fontanet, Lulla Opatowski, Pierre-Yves Boelle, and Simon
  Cauchemez.
\newblock Estimating the burden of {SARS-CoV-2} in {France}.
\newblock {\em Science}, 369(6500):208--211, 2020.

\bibitem{schertzer2018height}
Emmanuel Schertzer and Florian Simatos.
\newblock Height and contour processes of {Crump}-{Mode}-{Jagers} forests
  {(I)}: general distribution and scaling limits in the case of short edges.
\newblock {\em Electronic Journal of Probability}, 23:43pp., 2018.

\bibitem{sellke1983asymptotic}
Thomas Sellke.
\newblock On the asymptotic distribution of the size of a stochastic epidemic.
\newblock {\em Journal of Applied Probability}, 20:390--394, 1983.

\bibitem{sofonea_epidemiological_2020}
Mircea~T. Sofonea, Bastien Reyné, Baptiste Elie, Ramsès Djidjou-Demasse,
  Christian Selinger, Yannis Michalakis, and Samuel Alizon.
\newblock Memory is key in capturing covid-19 epidemiological dynamics.
\newblock {\em Epidemics}, 35:100459, 2021.

\bibitem{Taib1992}
Ziad Ta{\"{i}}b.
\newblock {\em Branching Processes and Neutral Evolution}, volume~93 of {\em
  Lecture Notes in Biomathematics}.
\newblock Springer Berlin Heidelberg, Berlin, Heidelberg, 1992.

\bibitem{thieme1985}
Horst~R. Thieme.
\newblock Renewal theorems for some mathematical models in epidemiology.
\newblock {\em Journal of Integral Equations}, 8(3):185--216, 1985.

\bibitem{Tindale2020}
Lauren~C Tindale, Jessica~E Stockdale, Michelle Coombe, Emma~S Garlock, Wing
  Yin~Venus Lau, Manu Saraswat, Louxin Zhang, Dongxuan Chen, Jacco Wallinga,
  and Caroline Colijn.
\newblock Evidence for transmission of covid-19 prior to symptom onset.
\newblock {\em eLife}, 9:e57149, 2020.

\bibitem{tran2008large}
Viet~Chi Tran.
\newblock Large population limit and time behaviour of a stochastic particle
  model describing an age-structured population.
\newblock {\em ESAIM: Probability and Statistics}, 12:345--386, 2008.

\bibitem{verity_estimates_2020}
Robert Verity, Lucy~C. Okell, Ilaria Dorigatti, Peter Winskill, Charles
  Whittaker, Natsuko Imai, Gina Cuomo-Dannenburg, Hayley Thompson, Patrick
  G.~T. Walker, Han Fu, Amy Dighe, Jamie~T. Griffin, Marc Baguelin, Sangeeta
  Bhatia, Adhiratha Boonyasiri, Anne Cori, Zulma Cucunub{\'a}, Rich FitzJohn,
  Katy Gaythorpe, Will Green, Arran Hamlet, Wes Hinsley, Daniel Laydon, Gemma
  Nedjati-Gilani, Steven Riley, Sabine van Elsland, Erik Volz, Haowei Wang,
  Yuanrong Wang, Xiaoyue Xi, Christl~A. Donnelly, Azra~C. Ghani, and Neil~M.
  Ferguson.
\newblock Estimates of the severity of coronavirus disease 2019: a model-based
  analysis.
\newblock {\em The Lancet Infectious Diseases}, 20:669--677, 2020.

\bibitem{wallinga07}
Jacco Wallinga and Marc Lipsitch.
\newblock How generation intervals shape the relationship between growth rates
  and reproductive numbers.
\newblock {\em Proceedings of the Royal Society B: Biological Sciences},
  274:599--604, 2007.

\bibitem{Wu2020}
Joseph~T. Wu, Kathy Leung, Mary Bushman, Nishant Kishore, Rene Niehus, Pablo~M.
  de~Salazar, Benjamin~J. Cowling, Marc Lipsitch, and Gabriel~M. Leung.
\newblock Estimating clinical severity of {COVID-19} from the transmission
  dynamics in {Wuhan}, {China}.
\newblock {\em Nature Medicine}, 26:506--510, 2020.

\end{thebibliography}

\appendix

\section{Appendix}

\subsection{Numerical simulation of the PDE} \label{A:numerical}

Computing the likelihood of our model requires to obtain an expression
for the solution $(n(t,a);\, t,a \ge 0)$ of the PDE \eqref{eq:McKVF}.
This equation was solved numerically using a backward difference scheme
based on the method of characteristics \cite{Kim2006}.

For $h > 0$, we approximate the value of  $(n(t,a);\, t, a \ge 0)$ 
on the lattice $\{ (i h, k h);\, i \le T^*, k \le A^* \}$ by the
array $(u(k, i);\, k,i)$ defined as follows:
\begin{gather*}
    \forall k \le T^*-1,\; i \le A^*-1,\quad u(k+1, i+1) = u(k, i) \\
    \forall i \le A^*,\quad u(0, i) = x_0 g(i h) \\
    \forall k \le T^*-1,\quad u(k+1, 0) = h \sum_{i = 1}^{A^*} \tau(ih) u(k+1,i).
\end{gather*}
Note that individuals with age larger than $h A^*$ are discarded. This 
maximal age was chosen so that individuals with age greater than $h A^*$
have negligible infection rate, and belong to the dead or recovered 
compartment with high probability.

\subsection{Expected number of entrances in a state} \label{A:entrance}

Theorem~\ref{thm:singlePop} provides an expression for the number of
individuals $n_i(t)$ in compartment $i$ at time $t$. As we want to fit
``incidence'' data, that is, the number of individuals that \emph{enter}
a given compartment, we need to derive an expression for this quantity.
To this aim, for two compartments $i,j \in \mathcal{S}$, let us write $i
\preceq j$ if an individual must visit $i$ before it gets to $j$. For
instance, in the admission model we have $C \preceq H \preceq I \preceq
D$. Then, the number of entrances $e_i(t, t+s)$ in $i$ between $t$ and
$t+s$ is given by
\begin{align} \label{eq:likAdmission}
    e_i(t, t+s) = n_i(t+s) - n_i(t) + \sum_{\substack{j \succeq i\\j \ne i}}
    n_j(t+s) - n_j(t).
\end{align}
The last term in the previous sum corresponds to the number of
individuals who leave compartment $i$ during $[t, t+s]$. Expression
\eqref{eq:likAdmission} can be readily used to derive the expected
number of entrance in $i$ from $(n(t, a);\, t,a \ge 0)$ and 
$(p(a,j);\, a \ge 0)$. 

It is interesting to note that \eqref{eq:likAdmission} only depends on
the distribution of the entrance time $T_i$ of $(X(a);\, a \ge 0)$ in
$i$, defined as:
\[
    T_i \coloneqq \inf \{ a \ge 0 : X(a) = i \}
\]
with the convention $\inf \emptyset = \infty$. To see this, one can write
\begin{align*}
    e_i(t,t+s) &= \int_0^\infty \big(n(t+s,a) - n(t,a)\big) 
                \Big(\sum_{j \succeq i} p_j(a) \Big) \diff a \\
               &= \int_0^\infty \big(n(t+s,a) - n(t,a)\big) 
               \P( T_i \le a)  \diff a.
\end{align*}
From an inference perspective, this is quite convenient since computing
$e_i(t,t+s)$ only requires to infer the distribution of $T_i$.

\subsection{Likelihood computation} \label{A:likExpression}

The daily incidence and prevalence data for France between March 18 and May
11 were taken from~\cite{donnee_france}. The days during this time period
are indexed by $\{ 1, 2, \dots, t_{\max} \}$, where day $1$ is March 18 and
day $t_{\max}$ is May 11.

For $i \in \{H, U, D\}$, we denote by $e^{\mathrm{obs}}_i(t)$ the
reported number of admissions to hospital, ICU, or the number of deaths
on day $t$, respectively. Moreover, for $i \in \{H, U, R\}$, we denote by
$n_i^{\mathrm{obs}}(t)$ the reported number of occupied beds in hospital, 
ICU, or the number of discharged patients on day $t$, respectively.
Let us further denote by
\[
    \pi(k; \lambda) = e^{-\lambda} \frac{\lambda^k}{k!}
\]
the probability mass function of a Poisson distribution with parameter
$\lambda$. 

Then, the likelihood of a parameter set $\theta$ under the admission
model is given by
\[
    L_{\mathrm{ad}}(\theta) =  \prod_{i \in \{H, U, D\}}
    \prod_{k = 1}^{t_{\max}} \pi\big(e^{\mathrm{obs}}_i(t_k); e_i(t_{k-1}, t_k)\big).
\]
The expected number of entrances in state $i$, $e_i(t_{k-1}, t_k)$, is
computed using \eqref{eq:likAdmission} with the one-dimensional marginals
of the admission model and the numerical approximation of
\eqref{eq:McKVF} described in Section~\ref{A:numerical}.

For the occupancy model, the likelihood of the parameter set $\theta$ is
given by
\begin{equation} \label{eq:likOcModel}
    L_{\mathrm{oc}}(\theta) =  \prod_{i \in \{H, U, D\}} 
    \prod_{k = 1}^{t_{\max}} \pi\big(e^{\mathrm{obs}}_i(t_k); e_i(t_{k-1}, t_k)\big) 
                            \times \prod_{i \in \{H, U, R\}} 
    \prod_{k = 1}^{t_{\max}} \pi\big(n^{\mathrm{obs}}_i(t_k); n_i(t_k)\big).
\end{equation}
Again, the value of $e_i(t_{k-1}, t_k)$ is computed using
\eqref{eq:likAdmission} and that of $n_i(t_k)$ using
\eqref{eq:prediction}, but using the one-dimensional marginals of the
occupancy model. Note that under this more complex model, there are two
pathways to hospital (critical and severe infection), two pathways to
death (with or without hospital admission), and three pathways to
hospital discharge (fast discharge, slow discharge, or discharge of ICU
patients). The predicted values for the number of individuals in each of
these compartments is obtained by summing over all pathways leading to
the corresponding state. 

For both models, we looked for the parameter set $\theta$ maximizing the likelihood.  
It was obtained using the \texttt{minimize} function of the Python
\texttt{scipy.optimize} module, using a Nelder-Mead algorithm. We
selected as initial point of the optimization algorithm a set of
parameters that were close to the existing estimates in the literature,
or which seemed realistic if such estimates did not exist.

\subsection{Best fitting prevalence curves under admission model} \label{A:poorFit}

Recall the admission model from Section~\ref{SS:admission}. By adding two
parameters to the model, one for the mean time between hospital admission
and discharge, the other for the mean time between ICU admission and
discharge, we can derive an expression for the likelihood of the
prevalence and incidence time series under the admission model, similar
to \eqref{eq:likOcModel}. The best-fitting values for these two
parameters were obtained by maximizing the likelihood with all other
parameters values fixed to those estimated in Table~\ref{T:admission}.
The corresponding model is displayed in Figure~\ref{F:poorFit}.

\begin{figure}[h]
    \includegraphics[width=\textwidth]{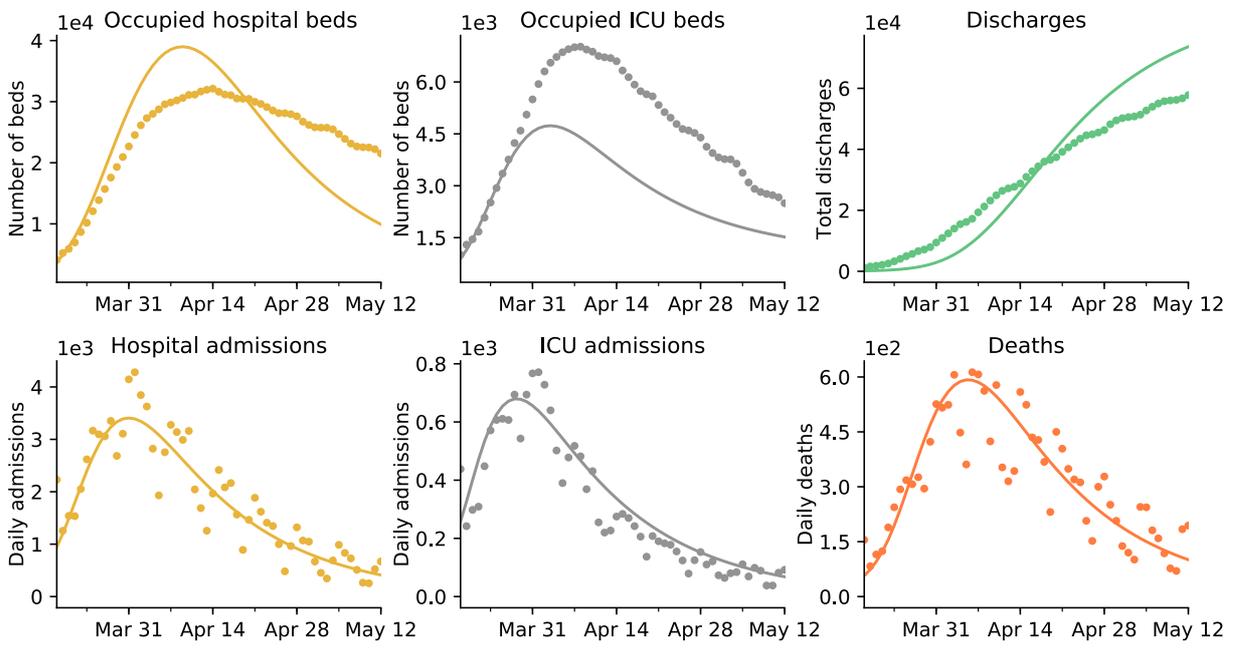}
    \caption{Best fit of the admission model for prevalence data.}
    \label{F:poorFit}
\end{figure}

\end{document}